\documentclass[11pt]{article}
\usepackage{fullpage}
\usepackage[top=0.9in,bottom=1.1in,left=1.1in,right=1.1in]{geometry}
\usepackage{graphicx}
\usepackage{amsmath,amssymb,graphicx,paralist,amsthm}
\RequirePackage{hyperref}
\usepackage{times}
\usepackage{bm}
\usepackage{natbib}
\usepackage{anyfontsize}

\usepackage[plain,noend]{algorithm2e}

\newtheorem{theorem}{Theorem}
\newtheorem{lemma}[theorem]{Lemma}

\newtheorem{proposition}{Proposition}
\newtheorem{condition}{Condition}

\DeclareMathOperator\tr{tr}
\DeclareMathOperator\diag{diag}
\DeclareMathOperator\col{col}
\DeclareMathOperator\cov{cov}
\DeclareMathOperator\sgn{sgn}
\DeclareMathOperator\var{var}
\DeclareMathOperator*{\argmin}{arg\,min}
\DeclareMathOperator*{\argmax}{arg\,max}

\newcommand{\A}{{A}}
\newcommand{\B}{{B}}
\newcommand{\D}{{D}}
\newcommand{\bE}{{E}}
\newcommand{\E}{{E}}
\newcommand{\cE}{\mathcal{E}}
\newcommand{\F}{{F}}
\newcommand{\cF}{\mathcal{F}}
\newcommand{\G}{{G}}
\newcommand{\bfH}{{H}}
\newcommand{\I}{{I}}
\newcommand{\J}{{J}}
\newcommand{\K}{{K}}
\newcommand{\cL}{\mathcal{L}}
\newcommand{\M}{{M}}
\newcommand{\cM}{\mathcal{M}}
\newcommand{\cN}{\mathcal{N}}
\newcommand{\bfP}{{P}}
\newcommand{\Q}{{Q}}
\newcommand{\R}{{R}}
\newcommand{\bR}{\mathbb{R}}
\newcommand{\bfS}{{S}}
\newcommand{\cS}{\mathcal{S}}
\newcommand{\bfT}{{T}}
\newcommand{\cT}{\mathcal{T}}
\newcommand{\cP}{\mathcal{P}}
\newcommand{\bP}{\mathbb{P}}
\newcommand{\U}{{U}}
\newcommand{\V}{{V}}
\newcommand{\bV}{\mathbb{V}}
\newcommand{\W}{{W}}
\newcommand{\X}{{X}}
\newcommand{\Y}{{Y}}
\newcommand{\Z}{{Z}}
\newcommand{\zero}{{0}}
\newcommand{\ii}{1}
\newcommand{\bfa}{{a}}
\newcommand{\bfb}{{b}}
\newcommand{\bfe}{{e}}

\newcommand{\bfu}{{u}}
\newcommand{\bfv}{{v}}
\newcommand{\bx}{{x}}
\newcommand{\bDelta}{{\Delta}}
\newcommand{\bGamma}{{\Gamma}}
\newcommand{\bLambda}{{\Lambda}}
\newcommand{\bOmega}{{\Omega}}
\newcommand{\bPi}{{\Pi}}
\newcommand{\bSigma}{{\Sigma}}
\newcommand{\bTheta}{{\Theta}}
\newcommand{\bepsilon}{{\epsilon}}
\newcommand{\bphi}{\phi}
\newcommand{\bmu}{\mu}
\newcommand{\bxi}{\xi}
\newcommand{\clr}{\text{clr}}

\newcommand{\ignore}[1]{}

\newenvironment{psmallmatrix}
  {\left(\begin{smallmatrix}}
  {\end{smallmatrix}\right)}

\def\T{{ \mathrm{\scriptscriptstyle T} }}
\def\v{{\varepsilon}}

\title{Principal component analysis for high-dimensional compositional data}

\author{Jingru Zhang \footnote{University of Pennsylvania, jingru.zhang@pennmedicine.upenn.edu} \qquad Wei Lin \footnote{Peking University, weilin@math.pku.edu.cn}}
\date{}                                       
\begin{document}

\maketitle

\begin{abstract}
Dimension reduction for high-dimensional compositional data plays an important role in many fields, where the principal component analysis of the basis covariance matrix is of scientific interest.
In practice, however, the basis variables are latent and rarely observed, and standard techniques of principal component analysis are inadequate for compositional data because of the simplex constraint. 
To address the challenging problem, we relate the principal subspace of the centered log-ratio compositional covariance to that of the basis covariance, and prove that the latter is approximately identifiable with the diverging dimensionality under some subspace sparsity assumption. 
The interesting blessing-of-dimensionality phenomenon enables us to propose the principal subspace estimation methods by using the sample centered log-ratio covariance. 
We also derive nonasymptotic error bounds for the subspace estimators, which exhibits a tradeoff between identification and estimation.
Moreover, we develop efficient proximal alternating direction method of multipliers algorithms to solve the nonconvex and nonsmooth optimization problems. Simulation results demonstrate that the proposed methods perform as well as the oracle methods with known basis. Their usefulness is illustrated through an analysis of word usage pattern for statisticians. 
\end{abstract}

\section{Introduction}
Compositional data arise naturally in a variety of applications, such as microbiome studies \citep{li2015microbiome}, chemical composition analysis \citep{van2019regional}, and text analysis \citep{blei2003latent}. In many contemporary datasets, the number of variables $p$ is often comparable to or even larger than the number of observations $n$. For example, in text analysis, tens of thousands of words could be collected, while the number of authors would be only tens or hundreds. 
One of the main tools for exploratory analysis of such high-dimensional compositional data is principal component analysis.
To facilitate interpretation, it is of scientific interest to find out sparse low-dimensional subspace that explains most of the variance.

The analysis of text corpora has attracted more and more attention in recent years \citep{cimiano2005learning,gretarsson2012topicnets}. For each author, it is easy to survey words in abstracts from some papers. After aligning these words to the reference vocabulary, one can quantify the relative abundances of words. Since only a part of papers are collected, this procedure only provides a relative, rather than absolute, measure of abundances. Therefore, the analysis usually starts from normalizing the observed data by the total number of counts and the resulting proportions fall into a class of high-dimensional compositional data.

%For the sake of description, we introduce some notations. 
Consider a text dataset with $p$ distinct words. We use $\W=(W_1,\ldots,W_p)^\T$ with $W_j>0$ for all $j$ to represent the abundances of the $p$ words, which is called the \emph{basis}. The observed \emph{compositional} data $\X=(X_1,\ldots,X_p)^\T$ are generated from $\W$ via 
\begin{equation*}
X_j = \frac{W_j}{\sum_{i=1}^pW_i},\quad j=1,\ldots,p.
\end{equation*}
An important problem in text analysis is how to select representative words from the whole vocabulary, which would highlight the word usage pattern for authors and benefit further data exploration such as clustering.
Ideally, the problem could be solved by analyzing sparse principal subspace on a covariance matrix of (transformed) basis, rather than a covariance matrix of (transformed) compositions \citep{pearson1897form}. In this paper, we focus on the sparse principal subspace of the \emph{basis covariance matrix} $\bOmega=(\omega_{ij})_{p\times p}$, which is defined by
\[
\omega_{ij} = \cov(Y_i,Y_j),
\]
where $Y_j=\log W_j$ is the log basis. 
In text analysis, the compositional data $\X$ are usually available, while the basis $\W$ are rarely observed. We need to seek a proxy for $\bOmega$ with the hope that its principal subspace could approximate that of $\bOmega$ well.

Principal component analysis for compositional data has been studied by many researchers. As pointed out by \cite{aitchison1983principal}, compositional data frequently display marked curvature owing to the simplex constraint. A direct application of principal component analysis to the raw compositions is unable to capture this nonlinear structure. This prompts the exploration in transformations for compositions. For example, \cite{aitchison1983principal} recommended transforming compositional data $\X$ to Euclidean data $\Z$ by the centered log-ratio transformation
\[
\Z  = (Z_1,\ldots,Z_p)^\T = \clr(\X) = \left(\log\frac{X_1}{g(\X)},\ldots,\log\frac{X_p}{g(\X)}\right)^\T,
\]
where $g(\bx)=(\Pi_{j=1}^px_j)^{1/p}$ is the geometric mean of a vector $\bx=(x_1,\ldots,x_p)^\T$, and defined the \emph{centered log-ratio covariance matrix} $\bGamma=(\gamma_{ij})_{p\times p}$ by
\[
    \gamma_{ij} = \cov(Z_i,Z_j).
\]
\cite{filzmoser2009principal} proposed to use the isometric log-ratio transformation, where the transformed data was full rank but cannot be interpreted directly. \cite{scealy2015robust} considered the power transformed compositional data and conducted principal component analysis in
a tangent space. 
Although these existing works on principal component analysis for compositional data have shown good performances under some practical scenarios, none of them provide a transparent interpretation in the sense of basis or theoretical analysis to ensure their effectiveness.
In addition, previous work does not consider the high-dimensional setting, which we are particularly interested in.

Some calculations show that $\bGamma$ is related to $\bOmega$ through the identity
\begin{equation}\label{eq:identity}
\bGamma=(\I-p^{-1}\J)\bOmega (\I-p^{-1}\J)\equiv \G\bOmega \G,
\end{equation}
where $\G=\I-p^{-1}\J$ with $\I$ being the $p\times p$ identity matrix and $\J$ being the $p\times p$ all-ones matrix. This nice property enables us to take $\bGamma$ as a proxy for $\bOmega$.
%By contrast, there does not exist such simple relationship between $\bOmega$ and covariance of other popular transformed compositions. 
However, owing to the  singularity of $\G$, $\bOmega$ is not uniquely determined by $\bGamma$, which would make a big difference in principal subspace between $\bOmega$ and $\bGamma$. 
To address the unidentifiability issue of $\bOmega$, \cite{cao2019large} worked under some sparsity assumption on $\bOmega$, and proved that the difference between $\bOmega$ and $\bGamma$ can vanish asymptotically. The sparsity assumption on the whole matrix $\bOmega$ seems too strong for principal component analysis, since only the low-dimensional principal subspace is of interest. 

In this paper, we only assume the principal subspace of $\bOmega$ is sparse. Based on that, we prove the difference in principal subspace between $\bGamma$ and $\bOmega$ vanishes as the dimension goes to infinity. The sparsity assumption on the principal subspace, on the other hand, has shown its necessity in principal component analysis for high-dimensional Euclidean data \citep{paul2007asymptotics, nadler2008finite, johnstone2009consistency}. For example, \cite{ma2013sparse} and \cite{cai2013sparse} considered sparse principal subspace estimation for the spiked covariance matrix. \cite{vu2013minimax} introduced two complementary notions of subspace sparsity, row sparsity and column sparsity, and analyzed sparse principal subspace estimation without Gaussian or spiked covariance assumptions.
Since the existing work on principal component analysis for high-dimensional data only provides the estimation error bound, one straightforward idea is separating the approximation and estimation processes. However, as illustrated later, the sparsity on the principal subspace of $\bOmega$ may not hold for $\bGamma$. Therefore, the idea of separation would lead to a fairly large estimation error bound.

By relating the principal subspace of $\bGamma$ to that of $\bOmega$, this paper bridges the gap between principal component analysis on compositions and principal component analysis on basis. The connection enables the principal subspace estimation to enjoy a direct interpretation in terms of the basis. Specifically, we adopt the subspace sparsity introduced in \cite{vu2013minimax}. Under the subspace sparsity assumption, we prove that the principal subspace of the basis covariance matrix is asymptotically identifiable. The error bound of the sparse principal subspace estimation consists of two terms: One accounts for the estimation error caused by finite samples; another comes from the approximation error caused by the proxy. 
We will see the dimension $p$ plays opposite roles in these two terms, which reveals an intriguing tradeoff between estimation and identification.
We emphasize our theoretical analysis is not a straightforward extension of existing principal component analysis methods for high-dimensional Euclidean data, since the idea of separating the approximation and estimation processes is infeasible.
Additionally, the sparse principal subspace estimation can be formulated as a nonconvex and nonsmooth optimization problem. To solve this problem, we propose an alternating direction method of multipliers \citep{boyd2011distributed} algorithm. Simulation studies suggest that the proposed methods outperform the methods based on other commonly used transformations. We illustrate our methods by analyzing a text dataset in order to identify representative words and research directions of statisticians.

%The remainder of the paper is organized as follows. We introduce the (sparse) principal subspace for compositional data in Section \ref{sec:setup} and present our theoretical analysis in Section \ref{sec:theory}. Section \ref{sec:admm} provides the alternating direction method of multipliers algorithm for the sparse principal subspace estimation. Sections \ref{sec:simulation} and \ref{sec:apply} examine the performance of the proposed methods through simulation studies and an application to text analysis, respectively. We conclude the paper with discussion in Section \ref{sec:discuss}.

\section{Methodology}\label{sec:setup}

\subsection{Setup}
We first introduce some notation. For any matrix $\A =(a_{ij})\in\mathbb{R}^{m\times n}$, let  $\|\A\|_2=\sigma_{\max}(\A)$, $\|\A\|_F=(\sum_{i,j}a_{ij}^2)^{1/2}$, $\|\A\|_{\max}=\max_{i,j}|a_{ij}|$ and $\|\A\|_*=\sum_{i=1}^{\min(m,n)}\sigma_i(\A)$, where $\{\sigma_i(\A)\}$ are singular values of $\A$ and $\sigma_{\max}(\A)$ represents the largest singular value. For two vectors or matrices $\A$ and $\B$ of compatible dimension, define the inner product $\langle \A,\B\rangle = \tr(\A^\T\B)$. Let $\I_{p\times p}$ be the $p\times p$ identity matrix, where the subscript $p\times p$ sometimes is omitted when it is clear from the context. Let $\ii_p$ be the $p$-dimensional all-ones vector and $\J=\ii_p\ii_p^\T$ be the $p\times p$ all-ones matrix. Denote by $\bV_{p,d}$ the class of $p\times d$ matrices with orthonormal columns. We use $\diag(x_1,\ldots,x_p)$ to denote the diagonal matrix with the diagonal elements $(x_1,\ldots,x_p)$. For an orthogonal projection matrix $\E$, we use $\E^\bot$ to denote $\I-\E$.

Recall that we are interested in the principal subspace of $\bOmega$, and its sample covariance matrix is 
\[
\bfS_Y = \frac{1}{n}\sum_{j=1}^n(\Y_j-\overline{\Y})(\Y_j-\overline{\Y})^\T,
\]
where $\overline{\Y}=\sum_{j=1}^n\Y_j/n$. However, owing to the unavailability of the basis, we are actually working on the sample covariance matrix of $\bGamma$:
\[
\bfS_{Z} = \frac{1}{n}\sum_{j=1}^n(\Z_j-\overline \Z)(\Z_j-\overline \Z)^\T,
\]
where $\Z_j=\G\Y_j$ and $\overline \Z=\sum_{j=1}^n\Z_j/n$. 
% The relationship between $\bOmega$ and $\bGamma$ (Equation (\ref{eq:identity})) also holds for $\bfS_Y$ and $\bfS_Z$, i.e.,
% \[
%     \bfS_Z = \left(\I-p^{-1}\J\right) \bfS_Y \left(\I-p^{-1}\J\right)= \G\bfS_Y\G.
% \]

A key question is under what conditions the principal subspace of $\bGamma$ approximates that of $\bGamma$ well. 
In general, they can be very different.
Recall that $\G = \I_{p\times p}-p^{-1}\J$. Let $\G = \R_0\D_0\R_0^\T$ be the spectral decomposition of $\G$, 
where $\D_0=\diag(0,1,\ldots,1)$ and $\R_0$ is orthonormal. Consider $\bOmega=\R_0\bLambda \R_0^\T$, where $\bLambda = \diag(\lambda_1,\ldots,\lambda_p)$ with $\lambda_1 > \lambda_2\geq\lambda_3\geq\ldots\geq\lambda_p\geq 0$. Then
\[
\bGamma = \G\bOmega \G = \R_0\D_0\bLambda \D_0\R_0^\T = \R_0\tilde{\bLambda}\R_0^\T,
\]
where $\tilde{\bLambda} =\D_0\bLambda \D_0=\diag(0,\lambda_2,\ldots,\lambda_p)$. Thus, in this case, it is impossible to recover the first principal subspace of $\bOmega$ from $\bGamma$.

\subsection{Principal subspace for compositional data}
Let $\col(\U)$ denote the span of the columns of $\U$. For $\U\in\bV_{p,d}$, the orthogonal projection matrix for $\col(\U)$ is $\U\U^\T$.

Consider the spectral decomposition of the basis covariance matrix $\bOmega$:
\[
\bOmega = \sum_{j=1}^p\lambda_j\bfv_j\bfv_j^\T= \V_0\bLambda_0\V_0^\T+\V_1\bLambda_1\V_1^\T
\]
where $\lambda_1\geq\lambda_2\geq\ldots\geq\lambda_p\geq 0$ are the eigenvalues, $\bfv_1,\ldots,\bfv_p\in \bR^p$ are the associated orthonormal eigenvectors, $\V_0=(\bfv_1,\ldots,\bfv_d)\in\bR^{p\times d}$, $\V_1=(\bfv_{d+1},\ldots,\bfv_p)\in\bR^{p\times (p-d)}$, $\bLambda_0=\text{diag}(\lambda_1,\ldots,\lambda_d)$, and $\bLambda_1=\text{diag}(\lambda_{d+1},\ldots,\lambda_p)$. 
The $d$-dimensional \emph{principal subspace} of $\bOmega$ is denoted by
\[
\cS_\bOmega = \text{span}\{\bfv_1,\ldots,\bfv_d\} = \col(\V_0)
\]
and the associated projection matrix is $\V_0\V_0^\T$. 
We assume that the eigengap satisfies $\lambda_d-\lambda_{d+1}>0$, which ensures the principal subspace $\cS_\bOmega$ is uniquely defined. 

Similarly, let the spectral decomposition of the centered log-ratio covariance matrix $\bGamma$ be
\[
\bGamma = \sum_{j=1}^pa_j\bfu_j\bfu_j^\T= \U_0\A_0\U_0^\T+\U_1\A_1\U_1^\T
\]
and the corresponding $d$-dimensional \emph{principal subspace} be
\[
\cS_\bGamma = \text{span}\{\bfu_1,\ldots,\bfu_d\} = \col(\U_0)
\]
with the associated projection matrix $\U_0\U_0^\T$.

%\subsection{Subspace distance}

Let $\cE$ and $\cF$ be two $d$-dimensional subspaces of $\bR^p$. Let $\E$ and $\F$ denote the corresponding projection matrices, and the singular values of $\E\F^\bot$ be $s_1,\ldots,s_d,0,\ldots,0$. To measure the difference between $\cE$ and $\cF$, we adopt the distance  
\[
\|\sin\bTheta(\cE,\cF)\|_F,
\]
where $\bTheta(\cE,\cF)=\diag(\arcsin(s_1),\ldots,\arcsin(s_d))$ is the angle operator between $\cE$ and $\cF$ \citep{stewart1990matrix}.
We will frequently use the following identity
\[
    \|\sin\bTheta(\cE,\cF)\|_F^2 = \sum_{j=1}^d s_j^2 = \|\E\F^{\perp}\|_F^2=\frac{1}{2}\|\E-\F\|_F^2.
\]
%Since $\|\E-\F\|_F^2/2\leq d$, we have $\|\sin\bTheta(\cE,\cF)\|_F\leq\sqrt{d}$.

\subsection{Sparse principal subspace for compositional data}
We impose structural constraints on $\cS_\bOmega$ for two reasons. First, since $\bGamma=\G\bOmega \G$ and $\G$ is singular, the basis covariance matrix $\bOmega$ and its principal subspace $\cS_\bOmega$ are unidentifiable in general. However, with additional constraints on $\V_0$,  we may be able to bound the difference between $\cS_\bOmega$ and $\cS_\bGamma$. Second, 
additional structural constraints are necessary for reducing the estimation error, since standard principal subspace estimation would be inconsistent under the high-dimensional setting. Specifically, we assume  $\cS_\bOmega$ to be sparse in the sense of \cite{vu2013minimax}. 

We first introduce the class of row sparse principal subspace.
For a $p\times d$ matrix $\A$, define the $(2,q)$-norm, $q\in[0,1]$, as the usual $l_q$ norm of the vector of row-wise $l_2$ norms of $\A$:
\[
\|\A\|_{2,q} = \|(\|\bfa_{1*}\|_2,\ldots,\|\bfa_{p*}\|_2)\|_q,
\]
where $\bfa_{j*}\in\bR^d$ denotes the $j$th row of $\A$. 
For $0\leq q\leq1$ and $d\leq R_q\leq d^{q/2}\times p^{1-q/2}$, $\cS_\bOmega$ is row sparse if it belongs to
\begin{equation*}
\cM_q(R_q)=
\begin{cases}
\{\col(\U): \U\in\bV_{p,d} \text{ and } \|\U\|_{2,q}^q\leq R_q\}, & \text{ if } 0<q\leq1 \\
\{\col(\U): \U\in\bV_{p,d} \text{ and } \|\U\|_{2,0}\leq R_0\}, & \text{ if } q=0.
\end{cases}
\end{equation*}
% Denote $\cP_q(R_q)$ as the class of distributions on $\Y_1,\ldots,\Y_n$ that satisfy (\ref{yi}) and $\cS_\bOmega\in\cM_q(R_q)$.
The constraint that $d\leq R_q\leq d^{q/2}p^{1-q/2}$ is owing to the fact that the columns are orthonormal for any $\U \in\bV_{p,d}$. We refer to  \cite{vu2013minimax} for more details.
The row sparsity defined above ensures that the subspace is generated by only a small subset of $p$ variables.
We estimate the row sparse $\cS_\bOmega$ by solving the following constrained optimization problem
\begin{align}\label{opt-row}
\text{maximize } & ~\langle \bfS_Z,\U\U^\T\rangle \\ \notag
\text{subject to } & ~\U\in\bV_{p,d} \\ \notag
& ~\|\U\|_{2,q}^q\leq R_q \quad (\text{or }\|\U\|_{2,0}\leq R_0 \text{ if } q=0).
\end{align}

As complementary, we also consider the class of column sparse principal subspace.
For a $p\times d$ matrix $\A$, define the $(*,q)$-norm, $q\in[0,1]$, as the maximal $l_q$ norm of its columns:
\[
\|\A\|_{*,q}=\max_{1\leq j\leq d}\|\bfa_{* j}\|_q,
\]
where $\bfa_{* j}\in\bR^p$ denotes the $j$th column of $\A$.
For $0\leq q\leq1$ and $1\leq R_q\leq p^{1-q/2}$, $\cS_\bOmega$ is column sparse if it belongs to
\begin{equation*}
\cM_q^*(R_q)=
\begin{cases}
\{\col(\U): \U\in\bV_{p,d} \text{ and } \|\U\|_{*,q}^q\leq R_q\}, & \text{ if } 0<q\leq1 \\
\{\col(\U): \U\in\bV_{p,d} \text{ and } \|\U\|_{*,0}\leq R_0\}, & \text{ if } q=0.
\end{cases}
\end{equation*}
% Denote $\cP_q^*(R_q)$ as the class of distributions on $\Y_1,\ldots,\Y_n$ that satisfy (\ref{yi}) and $\cS_\Omega\in\cM_q^*(R_q)$.
The column sparsity sets constraints on each column of $\U\in\bV_{p,d}$ in the usual sense of $l_q$ sparse, and the column sparse $\cS_\bOmega$ can be estimated by solving
\begin{align}\label{opt-col}
\text{maximize } & ~\langle \bfS_Z,\U\U^\T\rangle \\ \notag
\text{subject to } & ~ \U\in\bV_{p,d} \\ \notag
& ~ \|\U\|_{*,q}^q\leq R_q \quad (\text{or }\|\U\|_{*,0}\leq R_0 \text{ if } q=0).
\end{align}

\section{Theory}\label{sec:theory}
%In this section, we first analyze the identifiability of the sparse principal subspace $\cS_\bOmega$ in Section \ref{sec:identify}. Then, the classical principal subspace estimation result (i.e., no sparse constraint is applied) is discussed  in Section \ref{sec:pca}. Section \ref{sec:sppca} provides the upper error bounds on the sparse principal subspace estimation over the row sparse and column sparse classes, respectively. Lastly, we present an example to illustrate that separating the estimation and approximation processes is infeasible.

\subsection{Identifiability of the principal subspace}\label{sec:identify}

We use $a = O(b)$ to denote that $a$ and $b$ are of the same order, and $a=o(b)$ to denote that $a$ is of a smaller order than $b$.
We will repeatedly use the quantities defined below
\begin{equation}\label{eq:sigma}
\sigma_1^2=\frac{\lambda_1\lambda_{d+1}}{(\lambda_d-\lambda_{d+1})^2}, \quad\quad\sigma_2^2=\frac{\lambda_1^2}{(\lambda_d-\lambda_{d+1})^2},
\end{equation}
and
\begin{equation*}
c(q)=
\begin{cases}
\frac{2-q}{2(1-q)}\left\{\frac{2(1-q)}{q}\right\}^{q/(2-q)} & \text{if } q\in (0,1)\\
2 & \text{if } q=0,1.
\end{cases}
\end{equation*}

\begin{theorem}\label{th:identify}
If $\cS_\bOmega\in\cM_q(R_q)\cup\cM_q^*(R_q)$, $q\in [0,1]$, then 
\begin{equation*}\label{ineq:iden}
\|\sin\bTheta(\cS_\bOmega,\cS_\bGamma)\|_F^2 \leq \frac{9c(q)^2\sigma_2^2d^2R_q^{2/(2-q)}}{p}.
\end{equation*}
\end{theorem}

The proof of Theorem \ref{th:identify} is deferred to Appendix \ref{pf:identify}.
The basic intuition is as follows.
Assume that ${v}$ is an eigenvector of $\bOmega$ satisfying $\bOmega \bfv = \lambda \bfv$. If ${v}^\T\ii_p=0$,  $\bfv$ is also an eigenvector of $\bGamma$ since 
\[
\bGamma\bfv = (\I-p^{-1}\J)\bOmega(\I-p^{-1}\J)\bfv = (\I-p^{-1}\J)\lambda \bfv = \lambda \bfv.
\]
If $\bfv$ is sparse, then $\langle \bfv/\|\bfv\|_2, \ii_p/p^{1/2}\rangle=O(\|\bfv\|_0^{1/2}/p^{1/2})$ goes to zero in the high-dimensional regime. Hence the principal subspace of $\bOmega$ could be approximated by that of $\bGamma$.

Theorem \ref{th:identify} ensures that $\cS_\bOmega$ is approximately identifiable as long as $\sigma_2^2d^2R_q^{2/(2-q)}=o(p)$. Under this condition, the difference between $\cS_\bOmega$ and $\cS_\bGamma$ vanishes asymptotically, which allows us to estimate $\cS_\bOmega$ based on the sample covariance matrix of $\bGamma$.

\subsection{Standard principal subspace estimation}\label{sec:pca}
In this section, we analyze the standard principal subspace estimation obtained by the eigen decomposition of $\bfS_Z$. We assume that there exist i.i.d.\ random vectors $\bfT_1,\ldots,\bfT_n\in\bR^p$ with $\bE (\bfT_1)=\zero$ and $\var(\bfT_1)=\I_{p\times p}$, such that 
\begin{equation}\label{yi}
\Y_k = \bmu+\bOmega^{1/2}\bfT_k \quad\text{ and }\quad \|\bfT_k\|_{\psi_2}\leq 1
\end{equation}
for $k=1,\ldots,n$, where $\|\cdot\|_{\psi_2}$ is the sub-Gaussian norm \citep{vershynin2018high} defined by 
\[
\|\bfT\|_{\psi_2}=\sup_{{b}:\|{b}\|_2\leq 1}\inf\Big\{C>0:\bE\exp\Big|\frac{\langle \bfT,{b}\rangle}{C}\Big|^2 \leq 2\Big\}.
\]
Denote by $\cP_q(R_q)$ the class of distributions on $\Y_1,\ldots,\Y_n$ that satisfy (\ref{yi}) and $\cS_\bOmega\in\cM_q(R_q)$.
Denote by $\cP_q^*(R_q)$ the class of distributions on $\Y_1,\ldots,\Y_n$ that satisfy (\ref{yi}) and $\cS_\bOmega\in\cM_q^*(R_q)$.

% Without any structure assumption on $\bOmega$, the (in)consistency of estimation based on $\Y_k$ is analyzed in Theorem \ref{th:pca-Sy}. Since only $\Z_k$ are observed in practice, we further derive the upper bound for principal subspace estimation based on $\Z_k$ in Theorem \ref{th:pca-Sz}. The proof of the theorems can be found in Supplement.
% First,  we have the following result.
% \begin{theorem}\label{th:pca-Sy} 
% Let $\bfS_Y=\widehat{\V}\widehat{\bLambda}\widehat{\V}^\T$ be the eigen decomposition of $\bfS_Y$, where $\widehat{\bLambda}=\diag(\widehat{\lambda}_1,\ldots,\widehat{\lambda}_p)$ with $\widehat{\lambda}_1\geq\widehat\lambda_2\geq\ldots\geq\widehat\lambda_p\geq 0$ and the orthonormal eigenvector matrix $\widehat{\V}=[\widehat{\bfv}_1,\ldots,\widehat{\bfv}_p]$. Let $\widehat{\V}_0=[\widehat{\bfv}_1,\ldots,\widehat{\bfv}_d]$ and $\widehat{\cS}_Y=\col(\widehat \V_0)$. If $\Y_1,\ldots,\Y_n$ satisfy (\ref{yi}), then for any $\tau\in (0,1/2)$, with probability $1-2\tau$, we have 
% \begin{equation}\label{ineq:Sy}
% \|\sin\bTheta(\widehat{\cS}_Y,\cS_\bOmega)\|_F^2\leq d\zeta_1^2/{\delta^2},
% \end{equation}
% where  $c_0$ is an absolute positive constant. 
% \end{theorem}

%The following theorem concerns the performance of standard principal component analysis.
\begin{theorem}[Standard principal component analysis]\label{th:pca-Sz}
 Let $q\in [0,1]$. Let $\bfS_Z=\widehat{\V}\widehat{\bLambda}\widehat{\V}^\T$ be the eigen decomposition of $\bfS_Z$, where $\widehat{\bLambda}=\diag(\widehat{\lambda}_1,\ldots,\widehat{\lambda}_p)$ with $\widehat{\lambda}_1\geq\widehat\lambda_2\geq\ldots\geq\widehat\lambda_p\geq 0$ and the orthonormal eigenvector matrix $\widehat{\V}=(\widehat{\bfv}_1,\ldots,\widehat{\bfv}_p)$. Let $\widehat{\V}_0=(\widehat{\bfv}_1,\ldots,\widehat{\bfv}_d)$ and $\widehat{\cS}_Z=\col(\widehat \V_0)$. If $\Y_i \stackrel{iid}{\sim}\bP\in\cP_q(R_q)\cup\cP_q^*(R_q)$, $i=1,\ldots, n$, then for any $\tau\in (0,1/2)$, we have 
\begin{equation}\label{ineq:Sz}
\|\sin\bTheta(\widehat{\cS}_Z,\cS_\bOmega)\|_F^2\leq \frac{c_1}{\delta^2}\left(d\zeta_1^2\vee\frac{c(q)^2\lambda_1^2d^2R_q^{2/(2-q)}}{p}\right)
\end{equation}
with probability at least $1-2\tau$, where $\zeta_1 = c_2\lambda_1\max\left((p-\log\tau)^{1/2}/n^{1/2},(p-\log\tau)/n\right)$, $\delta=\max\left(\lambda_d-\lambda_{d+1}-\zeta_1,0\right)$, $c_1$ and $c_2$ are positive constants.
\end{theorem}

The proof of Theorem \ref{th:pca-Sz} is provided in Appendix \ref{pf:pca-Sz}. The upper bound in Theorem \ref{th:pca-Sz} contains two parts, where the first term represents the estimation error and the second term accounts for the approximation error. 
The estimation error term coincides with the existing results (for example, \cite{johnstone2009consistency}), although we work without the spiked assumption for the covariance matrix and aim to estimate the principal subspace rather than the principal eigenvector. 
The approximation error term is specific for compositional data, which demonstrates the benefit of the sparsity structure on the principal subspace.

In particular, when $p/n \ll 1$, with probability at least $1-2/n$, the order of upper bound in (\ref{ineq:Sz}) is 
\begin{equation*}
\sigma_2^2\left\{\frac{d(p+\log n)}{n}\vee \frac{d^2R_q^{2/(2-q)}}{p}\right\}.
\end{equation*}
Assuming $d=O(1)$ and $\sigma_2=O(1)$, we see that $\widehat\cS_Z$ is consistent if both $p/n$ \emph{and} $R_q^{2/(2-q)}/p$ converge to zero. 
In the high-dimensional setting, the approximation error still goes to zero asymptotically, while the estimation error cannot vanish. It indicates that standard principal component analysis does not fully utilize the sparsity structure, which inspires us to incorporate the sparsity information into the estimation procedure. Specifically, we consider the constrained estimators \eqref{opt-row} and \eqref{opt-col}. 

\subsection{Sparse principal subspace estimation}\label{sec:sppca}
Let 
\[
\v_n=2^{1/2}R_q^{1/2}\left(\frac{d+\log p}{n}\right)^{1/2-q/4}.
\]
We need the following regularity conditions.
\begin{condition}\label{condition1}
There exists constants $c_1,c_2,c_3,c_4>0$ such that
\begin{equation*}
\v_n\leq 1,
\end{equation*}
\begin{equation*}
c_1\lambda_1\left(1+\frac{2c(q)^2R_q^{2/(2-q)}d}{p}\right)\left(\frac{d+\log n}{n}\right)^{1/2}+ c_2\lambda_{d+1}(\log n)^{5/2}\v_n \leq \frac{1}{2}(\lambda_d-\lambda_{d+1}),
\end{equation*}
\begin{equation*}
c_3\lambda_{d+1}(\log n)^{5/2}\v_n \leq (\lambda_1\lambda_{d+1})^{1/2-q/4}(\lambda_d-\lambda_{d+1})^{q/2},
\end{equation*}
\begin{equation*}
c_4\lambda_{d+1}(\log n)^{5/2}\v_n^2 \leq (\lambda_1\lambda_{d+1})^{1-q/2}(\lambda_d-\lambda_{d+1})^{q-1}.
\end{equation*}
\end{condition}

\begin{condition}\label{condition2}
There exists a constant $c_5>0$ such that
\begin{equation*}
c(q)R_q^{1/(2-q)}(\lambda_1d)^{1/2}\leq c_5(\lambda_{d+1}p)^{1/2}.
\end{equation*}
\end{condition}

% \begin{remark}
Conditions \ref{condition1} and \ref{condition2} are required to bound the estimation error and approximation error, respectively. They are quite mild under the high-dimensional sparse subspace setting, allowing $R_q$, $d$ and $\lambda_j ~(j=1,d,d+1)$ to grow with the sample size $n$. For example, when $q=0$, let $p=n^a$, $R_0=n^b$, $d=n^c$, $\lambda_1=n^{\iota_1}$, $\lambda_d=n^{\iota_2}$ and $\lambda_{d+1}=n^{\iota_3}$ with $a>b>c>0$ and $\iota_1\geq\iota_2>\iota_3$. Then Conditions \ref{condition1} and \ref{condition2} can be satisfied for large enough $n$ if $b+c<1$, $\iota_1<\iota_2+(1-c)/2$ and $b+c+\iota_1<a+\iota_3$.
% \end{remark}

\begin{theorem}\label{th:upperbound1}
(Row sparse upper bound). Let $q\in[0,1]$, $\widehat \V_0$ be any solution of  (\ref{opt-row}), and $\widehat{\cS}_Z=\col(\widehat \V_0)$. If $ \Y_i\stackrel{iid}{\sim}\bP\in\cP_q(R_q)$, $i=1,\ldots,n$, and Conditions \ref{condition1} and \ref{condition2} hold, then
\begin{equation}\label{ineq:upper1}
\|\sin\bTheta(\widehat{\cS}_Z,\cS_\bOmega)\|_F^2\leq c_1\left\{R_q\sigma_1^{2-q}\left(\frac{d+\log p}{n}\right)^{1-q/2}\vee\left(\frac{c(q)^2\sigma_2^2d^2R_q^{2/(2-q)}}{p}\right)\right\}
\end{equation}
with probability at least $1-4/n-6\log n/n-1/p$. Here, $\sigma_1^2, \sigma_2^2$ are defined in (\ref{eq:sigma}) and $c_1$ is a positive constant.
\end{theorem}

The proof of this theorem is deferred to Appendix \ref{pf:upperbound}.
The upper bound in Theorem \ref{th:upperbound1} consists of two terms. The first one is the estimation error, which only depends on $\log p$ thanks to the enforcing of sparsity. By contrast, Theorem \ref{th:pca-Sz} suggests that the estimation error of standard principal component analysis  depends on $p$.  The second term is the approximation error, which decreases with increasing $p$. Thus, the sparsity structure benefits both estimation and approximation for compositional data.

According to \cite{vu2013minimax}, if the basis could be observed, the bound of the principal subspace estimation based on $\bfS_Y$ would be the first term on the right side of \eqref{ineq:upper1}, which is optimal up to a constant with an additional condition that 
\begin{equation}\label{ineq:addcond0}
R_q^{2/(2-q)}\leq p^\iota \text{ for some constant }\iota<1.
\end{equation}
%Now we discuss when our upper bound in (\ref{ineq:upper1}) can achieve  the optimal minimax rate. 
In addition, if
\begin{equation}\label{ineq:addcond1}
\frac{\sigma_2^2d^2R_q^{2/(2-q)}}{p} \lesssim R_q\sigma_1^{2-q}\left(\frac{d+\log p}{n}\right)^{1-q/2},
\end{equation}
the approximation error can be bounded by the estimation error. Hence the bound in Theorem \ref{th:upperbound1} would be optimal with additional conditions (\ref{ineq:addcond0}) and (\ref{ineq:addcond1}).

Let us further assume that $d$, $\sigma_1^2$ and $\sigma_2^2$ can be bounded by a universal constant. Then the upper bound in Theorem \ref{th:upperbound1} can be simplified to
\begin{align*}
\left(\frac{R_q^{2/(2-q)}\log p}{n}\right)^{1-q/2}\vee \left(\frac{R_q^{2/(2-q)}}{p}\right).
\end{align*}
The estimator $\widehat{\cS}_Z$ is consistent when $R_q^{2/(2-q)}\log p/n$ \emph{and} $R_q^{2/(2-q)}/p$ go to zero asymptotically. The additional condition (\ref{ineq:addcond1}) can be simplified to
\[
n^{1-q/2}R_q^{q/(2-q)}\lesssim p(\log p)^{1-q/2}.
\]
It suggests that for compositional data, to achieve the optimal rate, the dimension $p$ should be sufficiently large, which is rather different from the story for Euclidean data.

Notice that $\cM_q^*(R_q)\subseteq\cM_q(dR_q)$. With the similar technique, we can derive the estimation error for the column sparse class with $R_q$ replaced by $dR_q$, while the approximation error is unchanged.
\begin{theorem}\label{th:upperbound2}
(Column sparse upper bound). Let $q\in[0,1]$, $\widehat \V_0$ be any solution of (\ref{opt-col}), and $\widehat{\cS}_Z=\col(\widehat \V_0)$. Assume that $\Y_i\stackrel{iid}{\sim}\bP\in\cP_q^*(R_q)$, $i=1,\ldots,n$, Condition \ref{condition1} holds with $R_q$ replaced by $dR_q$, and Condition \ref{condition2} holds. Then
\begin{equation*}\label{ineq:upper2}
\|\sin\bTheta(\widehat{\cS}_Z,\cS_\bOmega)\|_F^2\leq c_1\left\{dR_q\sigma_1^{2-q}\left(\frac{d+\log p}{n}\right)^{1-q/2}\vee\left(\frac{c(q)^2\sigma_2^2d^2R_q^{2/(2-q)}}{p}\right)\right\}
\end{equation*}
with probability at least $1-4/n-6\log n/n-1/p$. Here, $\sigma_1^2,\sigma_2^2$ are defined in (\ref{eq:sigma}) and $c_1$ is a positive constant.
\end{theorem}

% \subsection{Infeasibility of separating the estimation and approximation processes}\label{sec:separate}
% Comparing the result based on $\bfS_Z$ with that based on $\bfS_Y$, we see that the estimation error part is unchanged and only an approximation error term is added, which coincides with the error bound for identification in Theorem \ref{th:identify}. 

\subsection{Infeasibility of separating the estimation and approximation processes} \label{sec:separate} 
%Our proofs use similar techniques developed in the seminal work \citep{vu2013minimax}. However, our results (Theorems \ref{th:upperbound1} and \ref{th:upperbound2}) are not straightforward extensions of the ones in \cite{vu2013minimax} as explained below. 
Comparing the error bound in Theorem \ref{th:upperbound1} with that in \cite{vu2013minimax}, we see that the estimation error part is unchanged and only an approximation error term is added, which coincides with the error bound for identification in Theorem \ref{th:identify}. 
One may ask whether we can separate the estimation and approximation processes and then derive the error bound by adding the two parts together by applying
\begin{equation}\label{eq:sep}
\|\sin\bTheta(\widehat{\cS}_Z,\cS_\bOmega)\|_F \leq \|\sin\bTheta(\widehat{\cS}_Z,\cS_\bGamma)\|_F + \|\sin\bTheta(\cS_\bGamma,\cS_\bOmega)\|_F.
\end{equation}
The separation, however, would give rise to a much worse bound. 

If applying \cite{vu2013minimax} to bound the first term on the right side of \eqref{eq:sep}, one needs to ensure the sparsity of $\cS_\bGamma$ instead of $\cS_\bOmega$. Although $\cS_\bOmega$ is under some sparsity assumption and $\cS_\bGamma$ and $\cS_\bOmega$ are close to each other, the sparsity of $\cS_\bGamma$ may not hold as illustrated by the following example.
Consider $d=1$ and $\bOmega=\text{diag}(\lambda_1,\ldots,\lambda_p)$ with $\lambda_1>\lambda_2=\ldots=\lambda_p$. In this case, $\cS_{\bOmega}=\text{span}\{{e}_1\}$ and $R_q=1$, where $e_1=(1,0,\ldots,0)^\T$. Since $\bGamma=\G\bOmega \G$, the first  eigenvector of $\bGamma$ is  
\[
\bxi_1 = \frac{1}{\sqrt{(p^2-p)}}(p-1,-1,\ldots,-1)^\T.
\]
Some calculations give $\|\bxi_1\|_q^q=O(p^{1-q})$, which implies that the sparsity factor $R_q$ of $\cS_\bGamma$ becomes $O(p^{1-q})$. Hence, directly applying \cite{vu2013minimax} to the first term on the right hand side of \eqref{eq:sep} yields a $O(p^{1-q}/n^{1-q/2})$ term, which is much worse than ours. 

\section{Proximal alternating direction method of multipliers algorithms}\label{sec:admm}
The sparse principal subspaces are estimated by solving the optimization problems (\ref{opt-row}) and (\ref{opt-col}), which are very challenging owing to the nonconvexity and nonsmoothness.  In the literature, many efforts have been made to the case $q=1$, for which the problems can be recast to certain convex problems \citep{vu2013fantope, wang2014tighten, gu2014sparse, qiu2019gradient, wang2020principal}. However, this observation  does not hold when $q<1$. In this paper, we adopt the proximal alternating direction method of multipliers algorithm proposed by \cite{zhang2019primal} to address these challenges.
\cite{zhang2019primal} studied the convergence of the algorithm for $q\in (0,1]$. 
Our numerical results in Section \ref{sec:numerical} demonstrate that this method also works very well for the case $q=0$, which is not considered in \cite{zhang2019primal}. 
%Next, we examine the cases of row and column sparsity separately.

%\subsection{Row sparsity}\label{sec:admm-row}
 We first consider the row sparse principal subspace estimator by solving the penalized version of  (\ref{opt-row}): 
\begin{align}\label{opt-row-new}
\text{minimize }& -\langle \bfS_Z,\U\U^\T\rangle + \alpha\|\V\|_{2,q}^q + \frac{\mu}{2}\|\Y\|_F^2\\ \notag
\text{subject to }& \U\in\bV_{p,d} \\ \notag
& \U-\V-\Y = 0.
\end{align}
Here, the constraint $\|\U\|_{2,q}^q\leq R_q$ in (\ref{opt-row}) is replaced by the penalty $\alpha\|\V\|_{2,q}^q$ and the hyperparameter $\mu$ is set to be large enough to ensure the closeness between $\U$ and $\V$. 
%The tuning of $\mu$ and $\alpha$ will be discussed in details in Section \ref{sec:tunepara}.
%Recall that $\bfS$ is a covariance matrix, and it is chosen to be $\bfS_Z$ in our case.
The augmented Lagrangian function for problem (\ref{opt-row-new}) is
\begin{align*}
\cL_\beta(\U,\V,\Y,\bLambda) = -\langle \bfS_Z,\U\U^\T\rangle + \alpha\|\V\|_{2,q}^q + \frac{\mu}{2}\|\Y\|_F^2  + \langle \V-\U+\Y,\bLambda\rangle +\frac{\beta}{2}\|\V-\U+\Y\|_F^2,
\end{align*}
where $\bLambda$ is the Lagrange multiplier, $\beta>0$ is a penalty hyperparameter.
%Notice that (\ref{opt-row-new}) is a nonconvex and nonsmooth optimization with manifold constraints. Here, we apply the method introduced in \cite{zhang2019primal} and use a variant of alternating direction method of multipliers which linearizes the objective function. 
We define the following approximation to the augmented Lagrangian function:
\begin{align*}
\widehat\cL_\beta^\U(\U;\widehat \U,\widehat \V,\widehat \Y,\bLambda) = & -\langle \bfS_Z,\widehat \U\widehat \U^\T\rangle +  \frac{\mu}{2}\|\widehat \Y\|_F^2 -2\langle \bfS_Z\widehat \U,\U-\widehat \U\rangle + \alpha\|\widehat \V\|_{2,q}^q \\ 
&+ \langle \widehat \V-\U+\widehat \Y,\bLambda\rangle +\frac{\beta}{2}\|\widehat \V-\U+\widehat \Y\|_F^2. 
\end{align*}
The linearized proximal alternating direction method of multipliers algorithm is described in Algorithm \ref{admm1}.
\vspace{-0.6cm}
\begin{algorithm}\label{admm1}
\caption{Linearized proximal alternating direction method of multipliers for row sparsity.}
\begin{tabbing}
  \qquad \enspace Input: Initial values $\U^0,\V^0,\Y^0,\bLambda^0$ and hyperparameters $\alpha,\beta,\mu,\rho$ \\
  \qquad \enspace For $k=0$ to $k=K-1$ \\
  \qquad\qquad $\U^{k+1} = \arg\min\limits_{\U\in\bV_{p,d}} \widehat\cL_\beta^\U(\U;\U^k,\V^k,\Y^k,\bLambda^k)+\frac{\rho}{2}\|\U-\U^k\|_F^2$ \\
  \qquad\qquad $\V^{k+1} = \arg\min\limits_\V\cL_\beta(\U^{k+1},\V,\Y^k,\bLambda^k)+\frac{\rho}{2}\|\V-\V^k\|_F^2$ \\
  \qquad\qquad $\Y^{k+1} = \frac{1}{\mu+\beta}\left[\beta(\U^{k+1}-\V^{k+1})-\bLambda^{k+1}\right]$\\
  \qquad\qquad $\bLambda^{k+1} = \bLambda^k+\beta(\V^{k+1}-\U^{k+1}+\Y^{k+1})$ \\
  \qquad \enspace Output $\U^K$, $\V^K$
\end{tabbing}
\end{algorithm}

In $\U$-update of Algorithm \ref{admm1}, the subproblem to be solved is
\begin{equation}\label{sub1}
\U^{k+1} = \argmin\limits_{\U\in\bV_{p,d}} -\langle 2\A,\U\rangle = \argmin\limits_{\U\in\bV_{p,d}}\|\A-\U\|_F^2,
\end{equation}
where $\A = \bfS_Z\U^k+\frac{1}{2}(\bLambda^k+\beta \V^k+\beta \Y^k+\rho \U^k)$. Let $\A=\Q\D\bfP^\T$ be the SVD decomposition of $\A$. Then the solution of (\ref{sub1}) is  $\U^{k+1}=\Q\bfP^\T$.
$V$-update of Algorithm \ref{admm1} consists of $p$ decoupled subproblems:
\begin{align*}
& \bfv_{i*} = \arg\min_{\bfv_{i*}}\frac{\beta+\rho}{2}\|\bfv_{i*}\|_2^2+\alpha\|\bfv_{i*}\|_2^q+\bfv_{i*}^\T\bfb_{i*},~i=1,\ldots,p,\text{ if }q\in(0,1], \\
& \bfv_{i*} = \arg\min_{\bfv_{i*}}\frac{\beta+\rho}{2}\|\bfv_{i*}\|_2^2+\alpha \text{I}(\|\bfv_{i*}\|_2\neq 0)+\bfv_{i*}^\T\bfb_{i*},~i=1,\ldots,p,\text{ if }q=0, 
\end{align*}
where $\bfb_{i*}^\T$ is the $i$th row of the matrix $\B=\bLambda^k+\beta(\Y^k-\U^{k+1})-\rho \V^k$ and $\text{I}(\cdot)$ is an indicator function. 
After some calculations, we obtain the solutions for $q=1$ and $q=0$.
\begin{proposition}
For $q=1$, the solution to $V$-update of Algorithm \ref{admm1} is $\bfv_{i*} = \min(\alpha-\|\bfb_{i*}\|_2,0)\bfb_{i*}/\{(\beta+\rho)\|\bfb_{i*}\|_2\}$.
For $q=0$, the solution to $V$-update of Algorithm \ref{admm1} is
$\bfv_{i*} = -\text{I}(\|\bfb_{i*}\|_2^2>2\alpha(\beta+\rho))\bfb_{i*}/(\beta+\rho)$.
\end{proposition}
The subproblems also have closed-form solutions for some $q\in(0,1)$, which is discussed in details in Appendix \ref{supp:algo1}.

%\subsection{Column sparsity}\label{sec:admm-col}
Analogously, we can estimate the column sparse principal subspace by a similar alternating direction method of multipliers algorithm, and provide the details in Appendix \ref{supp:algo2}.

%\subsection{Tuning the hyperparameters}\label{sec:tunepara}
In  (\ref{opt-row-new}), there are two important hyperparameters $\alpha$ and $\mu$. The former  controls the sparsity of $\V$ and the latter controls the  difference between $\U$ and $\V$. 
The output $\U^K$ is orthonormal and $\V^K$ is sparse but not necessarily orthonormal. We set  $\mu$ to be large enough  such that $\V^K$ is nearly orthonormal, and take $\V^K$ as the sparse subspace estimator.  It is found that $\mu=1000$ works very well in our experiments (See details in Appendix \ref{supp:mu}).
 
The hyperparameter $\alpha$ is chosen through 5-fold cross-validation. 
Denote by ${\V}^{(-u)}(\alpha)$ the estimate based on the training data excluding the $u$th fold, and $\bfS_Z^{(u)}$ the sample covariance matrix on the test data including only the $u$th fold ($u=1,\ldots,5$). 
We choose the optimal value $\widehat\alpha$ which maximizes the sum of first $d$ eigenvalues
\[
\widehat\alpha = \argmax_\alpha\sum_{u=1}^5\langle \bfS_Z^{(u)},\V^{(-u)}(\alpha)(\V^{(-u)}(\alpha))^\T\rangle.
\]
With the optimal $\widehat\alpha$, we then compute the sparse subspace estimator based on the whole dataset as our final result.

%The other hyperparameters are chosen as follows. 
We set  $\Y^0=\bLambda^0=\zero_{p\times d}$, and find that setting $\U^0$ and $\V^0$ to be the first $d$ principal eigenvectors of $\bfS_Z$ works very well. 
According to the theoretical analysis in \cite{zhang2019primal}, we set $\beta=5.8\|\bfS_Z\|_2$ and $\rho=6.14\|\bfS_Z\|_2$.

\section{Numerical studies}\label{sec:numerical}
\subsection{Simulations}\label{sec:simulation}
In this paper, we focus on the centered log-ratio transformed compositions. In principle, the sparse principal subspace estimation can also be applied to other nonlinear transformations. Specifically, we consider the following transformations for comparison and include the oracle method as a baseline.
\begin{enumerate}
\item[(i)] Oracle: assuming that the basis could be observed and applying the algorithms to the sample covariance matrix of the log-basis $\Y$.
\item[(ii)] Log: applying the algorithms to the sample covariance matrix of the log-transformed compositions $\log\X$.
\item[(iii)] Raw: applying the algorithms to the sample covariance matrix of the raw compositions $\X$.
\item[(iv)] Power: applying the algorithms to the sample covariance matrix of the power transformed compositions proposed by \cite{scealy2015robust}.
\end{enumerate}

We set $p=500$, $d=5$, and consider $n=250$, $500$, and $1000$. We generate the log-basis vectors $\Y_k \in\mathbb{R}^p$ ($k=1,\ldots,n$) in two ways:
\begin{enumerate}
\item[(i)] multivariate normal distribution, $\Y_k\stackrel{iid}{\sim}\cN_p(\bmu,\bOmega)$;
\item[(ii)] multivariate gamma distribution, $\Y_k=\bmu+\F\U_k/\sqrt{10}$, where $\F\F^\T=\bOmega$ and the components of $\U_k$ are independent gamma variables with shape parameter 10 and scale parameter 1.
\end{enumerate}
The abundances $\W_k=(W_{k1},\ldots,W_{kp})^\T$ can be obtained by $W_{kj}=\exp({Y_{kj}})$.
Then $\X_k=(X_{k1},\ldots,X_{kp})^\T$ with $X_{kj}=W_{kj}/\sum_{i=1}^pW_{ki}$, and $\Z_k=(Z_{k1},\ldots,Z_{kp})^\T$ with $Z_{kj}=\log(X_{kj}/g(\X_k))$. 

In both cases, we take the components of $\bmu$ randomly from the uniform distribution on $[0,10]$.
For the basis covariance matrix $\bOmega$, we consider two scenarios:
\begin{enumerate}
\item[(i)] Row sparsity, $\cS_\bOmega\in\cM_0(R_0)$ with $R_0=10$. Define the sparse principal subspace $\V$ as $\V=\begin{psmallmatrix}
\V^{(1)} \\
{0}
\end{psmallmatrix}$, where $\V^{(1)}$ is a ${R_0\times d}$ matrix with orthonormal columns generated by sampling its entries from a standard Gaussian distribution and then orthonormalizing them;
\item[(ii)] Column sparsity, $\cS_\bOmega\in\cM_0^*(R_0)$ with $R_0=10$. Define the sparse principal subspace $\V$ as $\V=\begin{psmallmatrix}
\V^{(1)}  \\
{0} 
\end{psmallmatrix}$ with a block diagonal matrix $\V^{(1)}=\text{blkdiag}(\V^{(11)},\V^{22})$, where $\V^{(11)}$ is a ${R_0\times \lceil d/2\rceil}$ matrix with orthonormal columns generated by sampling its entries from a standard Gaussian distribution and then orthonormalizing them, and $\V^{(22)}$ is a ${R_0\times \lfloor d/2\rfloor}$ matrix generated similarly.
\end{enumerate}
We sample a $p\times p$ matrix $\K$ from a Wishart distribution with $p+10$ degrees of freedom and scale matrix $\I/p$. Let $\lambda_{d+1}=\|(\I-\V\V^\T)\K(\I-\V\V^\T)\|_2$ and $\bOmega=\V\D\V^\T+(\I-\V\V^\T)\K(\I-\V\V^\T)$, where $\D=\diag(\lambda_1,\ldots,\lambda_d)$ with $\lambda_i = \left\{3.6-2(i-1)/(d-1)\right\}\lambda_{d+1}$, $i=1,\ldots,d$.
%so that 
%\[
%\sigma_1^2=\frac{\lambda_1\lambda_{d+1}}{(\lambda_d-\lambda_{d+1})^2} = 10, \quad \sigma_2^2=\frac{\lambda_1^2}{(\lambda_d-\lambda_{d+1})^2} = 36.
%\]

\begin{table}
\def~{\hphantom{0}}
\caption{Comparisons of mean squared error (standard errors) with $p=500$ and a normal log-basis distribution based on 100 simulations}{
%\vspace{0.2cm}
\begin{tabular}{@{}llcccc@{}}
         & & \multicolumn{2}{c}{Row sparsity} & \multicolumn{2}{c}{Column sparsity}\\ 
         & Method & $q=0$ & $q=1$ & $q=0$ & $q=1$ \\ [5pt]
$n=250$ & Oracle & 0.016 (0.0005)  &  0.019 (0.0005) &  0.124 (0.0045) & 0.088 (0.0023)\\
         & Proposed & 0.017 (0.0006)  &  0.019 (0.0005) & 0.124 (0.0046) &  0.088 (0.0023)\\
         & Log & 1.650 (0.0879)  &  1.238 (0.0513) & 1.138 (0.0970) &  0.943 (0.0359)\\
         & Raw & 2.515 (0.0109)  &  2.506 (0.0153) & 2.500 (0.0000)  &  2.500 (0.0000)\\ 
         & Power & 2.489 (0.0074) & 2.474 (0.0120) & 2.500 (0.0000)  &  2.500 (0.0000)\\ 
$n=500$ & Oracle & 0.008 (0.0003) & 0.010 (0.0003) & 0.042 (0.0017) &  0.045 (0.0009) \\
         & Proposed & 0.008 (0.0003) & 0.011 (0.0003)  & 0.043 (0.0016) &  0.045 (0.0009) \\
         & Log & 1.318 (0.0603)  &  1.202 (0.0263)  &  0.862 (0.0741) &  0.970 (0.0238) \\
         & Raw & 2.493 (0.0050)  &  2.482 (0.0068)  &  2.500 (0.0000) & 2.500 (0.0000) \\ 
         & Power & 2.500 (0.0000) &  2.486 (0.0058) &  2.500 (0.0000) & 2.500 (0.0000) \\ 
$n=1000$ & Oracle  & 0.004 (0.0001) &  0.005 (0.0002)  &  0.018 (0.0007) &   0.024 (0.0005)\\
    & Proposed  & 0.004 (0.0001)  & 0.006 (0.0002)  &  0.018 (0.0008) &   0.024 (0.0005)\\
    & Log    &  1.018 (0.0534)  & 1.128 (0.0211)  &  0.731 (0.0575) & 0.979 (0.0170)\\
    & Raw   &  2.482 (0.0115) &  2.459 (0.0136)  &  2.500 (0.0000) & 2.500 (0.0000)\\       & Power & 2.492 (0.0053) & 2.482 (0.0071)  &  2.500 (0.0000) & 2.500 (0.0000)\\
\end{tabular}}\label{tab:normal}
\end{table}

\begin{table}
\def~{\hphantom{0}}
\caption{Comparisons of mean squared error (standard errors) with $p = 500$ and a gamma log-basis distribution based on 100 simulations}{
%\vspace{0.2cm}
\begin{tabular}{llcccc}
& & \multicolumn{2}{c}{Row sparsity} & \multicolumn{2}{c}{Column sparsity}\\ 
         & Method & $q=0$ & $q=1$ & $q=0$ & $q=1$ \\ [5pt]
$n=250$ & Oracle   &   0.016 (0.0005)  &  0.019 (0.0006)  &   0.121 (0.0052) &  0.087 (0.0017)\\
         & Proposed &   0.016 (0.0005)  &  0.019 (0.0006)  &   0.125 (0.0057) &  0.088 (0.0017)\\
         & Log       &  2.181 (0.0166)  &  2.102 (0.0278)  &   1.975 (0.0200) &  1.710 (0.0562)\\
         & Raw      &   3.690 (0.0752)  &  3.715 (0.0749)   &  2.525 (0.0120) &  2.550 (0.0223)\\ 
         & Power & 4.042 (0.0794) & 4.071 (0.0768) & 3.173 (0.0855) & 3.338 (0.0874) \\
$n=500$        & Oracle    & 0.008 (0.0002)  & 0.010 (0.0003)  &  0.048   (0.0020)  &  0.044 (0.0009)\\
    &Proposed  & 0.008 (0.0003) &  0.010 (0.0003)  &  0.048 (0.0020) &   0.045 (0.0009)\\
    &Log        &1.678 (0.0095)  & 1.645 (0.0139)  &  1.414 (0.0129)  &  1.466 (0.0308)\\
    &Raw       & 3.671 (0.0673)  & 3.737 (0.0675)  & 2.534 (0.0197)  &  2.610 (0.0360)\\
    & Power   & 4.129 (0.0738) & 4.138 (0.0746) & 3.254 (0.0851) & 3.394 (0.0858) \\
$n=1000$ & Oracle   &   0.004 (0.0001)  & 0.005 (0.0002)  &  0.019 (0.0008) &  0.025 (0.0005)\\
    & Proposed &   0.004 (0.0001) &  0.006 (0.0002)  &  0.020 (0.0008) &  0.025 (0.0005)\\
    & Log     &    1.359 (0.0053)  & 1.344 (0.0063)  &  1.136 (0.0097) & 1.252 (0.0158)\\
    & Raw    &     3.554 (0.0709)  & 3.637 (0.0745)  &  2.513 (0.0114)  &  2.575 (0.0247)\\
    & Power & 3.960 (0.0782) & 4.066 (0.0772) & 3.328 (0.0817) & 3.480 (0.0852)\\
\end{tabular}}\label{tab:gamma}
\end{table}

Let $\alpha=\exp(a_0)$ with $a_0\in\{-1.5,-1,\ldots,2.5,3\}$ for row sparsity and $a_0\in\{0.5,1,\ldots,$ $4.5,5\}$ for column sparsity. We select an optimal $\alpha$ by the 5-fold cross-validation.
Both $q=0$ and $q=1$ are adopted for the row and column sparsity. 
For each scenario, we repeat 100 simulations. 
 We use the squared distance between the estimator and the true subspace to measure the performance of our proposed approaches in comparison with other methods. 
Tables \ref{tab:normal} and \ref{tab:gamma} summarize the results of various methods with normal and gamma log-basis distributions, respectively. 
The proposed methods are not sensitive to the type of the log-basis distribution.
They perform as well as the oracle ones, which give much less errors than the other methods under all scenarios. As the sample size increases, the errors of our methods and the oracle methods decrease monotonically. However, we do not see the clear monotonicity for the Log, Raw, and Power methods. 

\subsection{Application to text data for statisticians}\label{sec:apply}
We illustrate the proposed methods by applying them to a text dataset for statisticians. The dataset was collected and analyzed by \cite{ji2016coauthorship}, and consists of 3607 authors and 3248 papers published in AoS, JASA, JRSS-B and Biometrika from
2003 to the first half of 2012. We first extract abstracts of the papers and obtain 12,462 distinct words. Then we prune the vocabulary by stemming each term to its root, removing function words, and removing terms that appear in less than 50 papers. After the cleaning, the total vocabulary size is 580. In our study, we focus on the 236 authors in the giant component of the coauthorship network, and obtain a $236\times 580$ count matrix, where the $(i,j)$ element denotes the count of the $j$th word used by the $i$th author. 
%For this sparse matrix, we treat the zero counts as sampling zeros and apply the Dirichlet-multinomial model \wl{references} to obtain the compositional data. 
We replace zero counts with 0.05 and transform the count data into compositions.

\begin{table}
\def~{\hphantom{0}}
\caption{Results under the column sparsity with $q=0$. Denote the proposed, Log, Raw and Power methods by M1, M2, M3 and M4, respectively. Denote by PC1 and PC2 the first and second principal components}{
%\vspace{0.2cm}
\begin{tabular}{@{}llrrrrlrrrr@{}}
     & Word & M1 & M2 & M3 & M4 & Word & M1 & M2 & M3 & M4 \\ [5pt]
 PC1 & active &  & -0.35 &  &  & nonstationary &  & -0.32 &  &  \\
     & autoregression &  & -0.32 &  &  & penalty &  &  & 0.25 &  \\
     & chain &  & -0.32 &  &  & process &  &  & -0.38 &  \\
     & covariance & 0.37 &  & 0.27 & 0.53 & regression & 0.38 & 0.38 & 0.38 & 0.69 \\
     & coefficient & 0.31 &  & 0.22 &  & select &  &  & 0.46 &  \\
     & equal & 0.32 &  &  &  & semiparametric & 0.41 & 0.29 &  & 0.21 \\
     & group &  &  & 0.22 &  & space &  &  & -0.28 &  \\
     & lasso &  &  & 0.15 &  & subspace &  & -0.32 &  &  \\
     & likelihood & 0.37 &  &  & 0.24 & test &  &  & -0.12 &  \\
     & linear & 0.32 &  & 0.18 & 0.33 & time &  &  & -0.34 &  \\
     & movement &  & -0.32 &  &  & volatilization &  & -0.37 &  &  \\
     & nonparametric & 0.33 &  &  & 0.20 &  &  &  &  &  \\ [5pt]
 PC2 & adaptive & -0.22 & -0.18 &  &  & normal &  &  & 0.12 & \\
     & baseline & 0.22 & 0.19 &  &  & number &  &  &  & -0.34 \\
     & cancer & 0.22 & 0.18 &  &  & optimize & -0.26 & -0.21 &  & -0.19 \\
     & censor & 0.33 & 0.28 &  &  & oracle & -0.20 & -0.17 &  & \\
     & classify &  &  & -0.31 &  & penalty &  & -0.16 &  & \\
     & clinic & 0.19 & 0.17 &  &  & posterior &  & 0.12 &  & \\
     & cluster &  &  & -0.26 &  & predict &  &  & -0.28 & \\
     & covariance &  &  & 0.23 &  & proportion &  & 0.16 &  & \\
     & converge & -0.17 & -0.13 &  &  & regular &  & -0.14 &  & \\
     & dimension & -0.28 & -0.20 &  & -0.48 & select & -0.22 & -0.18 &  & -0.28 \\
     & disease &  & 0.15 &  &  & semiparametric &  & 0.09 & 0.16 & \\
     & generalize &  &  &  & -0.25 & smooth & -0.23 & -0.17 &  & \\
     & hazard & 0.31 & 0.26 &  &  & space &  &  & -0.22 & \\
     & high & -0.26 & -0.19 &  & -0.38 & sparse & -0.30 & -0.26 &  & \\
     & inference &  & 0.17 &  &  & statistic &  &  & 0.16 & -0.22 \\
     & lasso &  & -0.15 &  &  & survive & 0.32 & 0.27 &  & \\
     & likelihood &  &  & 0.39 &  & test &  &  & 0.61 & \\
     & maximal &  & 0.16 & 0.12 &  & theory &  &  &  & -0.23 \\
     & measure &  &  &  & 0.18 & time &  & 0.20 & 0.14 & 0.41 \\
     & minimax &  & -0.13 &  &  & treat &  & 0.16 &  & \\
     & missing & 0.19 & 0.18 &  &  & under &  &  &  & -0.15 \\
     & noise &  & -0.16 &  &  & weight &  & 0.15 &  & \\
     & nonparametric &  &  & 0.14 &  &  &  &  &  & \\ 
\end{tabular}}\label{tab:appcol0}
\end{table}

The proposed methods and some commonly used approaches (Log, Raw, Power) are applied to estimate the principal subspace. Common words are selected under the row sparsity, while column sparsity gives $d$ sets of words in the first $d$ eigenvectors, respectively. Here we set $d=2$.
We only provide the results with $q=0$ in the main text. The results with $q=1$ are similar and deferred to Appendix \ref{supp:real}.

%We first take a look at the selected words from the row sparse subspace in Table \ref{tab:app1}. We see that the results of the proposed and Log methods are similar. 
%Both of the two methods select frequently used words related to the type of method, i.e., whether the method is nonparametric (e.g., ``likelihood'', ``nonparametric'', ``semiparametric''), and the type of real application, i.e., whether it is biomedical research (e.g., ``cancer'', ``censor'', ``hazard'', ``survive''). Almost no inessential word is selected by the proposed and Log methods. 
%However, the results of Raw and Power methods cannot reflect characteristics of biomedical research well. For example, among the words selected by Raw and Power, only ``cancer'' selected by Raw and ``survive'' selected by Power are  related to biomedical research.
%Moreover, some words chosen by Raw and Power methods cannot give useful information, such as, ``number'', ``theory'' and ``under''. 

Table \ref{tab:appcol0} lists the nonzero loadings of the first two components for the four approaches. Take a look at the results of the first principal component. We observe that the selected words can be used to distinguish whether a author focuses on fundamental problems or specific areas, and all these methods select the critical words well. The proposed and Power methods only select the words related to the fundamental problems and put positive loadings to them. The Log and Raw methods select some representing fundamental research directions and some representing specific areas, and the signs of the loadings of these two kinds of words are opposite. 

For the results in the second principal component, the selected words place emphasis on whether one pays attention to biomedical applications or not. Among the words selected by the proposed method, those related to biostatistics (``baseline'', ``cancer'', ``censor'', ``clinic'', ``hazard'', ``missing'', ``survive'') have positive loadings, and others are given negative loadings. Although the Log method can also identify both biostatistics and nonbiostatistics directions, it does not separate the two types very well, since it puts positive loadings to some nonbiomedical words (e.g., ``inference'', ``maximal'', ``semiparametric'') besides the biomedical words.
The Raw and Power methods cannot identify biostatistics, since those selected by them do not reflect any biomedical characteristics. 
%Additionally, Raw and Power methods give some inessential terms such as ``group'', ``number'' and ``under''.

\begin{figure}[htb!]
   \centering
   \includegraphics[width=0.95\textwidth]{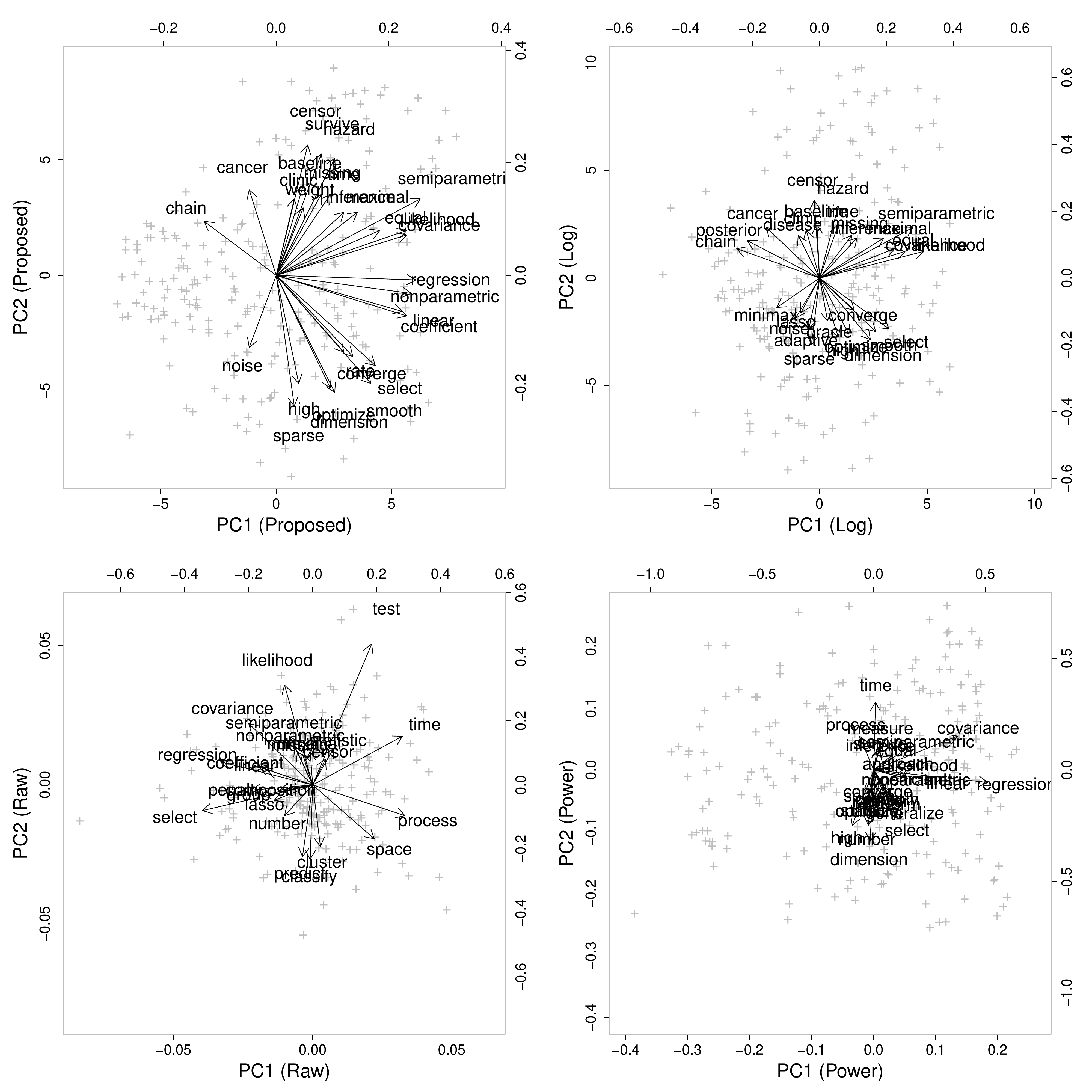}
   \caption{Biplots of the first two principal components for the proposed, Log, Raw and Power methods under the row sparsity with $q=0$.}\label{fig:realrow0}
\end{figure}

The biplots of the first two principal components for various methods are shown in Fig.\ \ref{fig:realrow0}. The plots for the proposed and Log methods are similar with the first principal component distinguishing fundamental and specific research topics and the second principal component separating biostatistics and nonbiostatistics. However, a close look at the two biplots reveals different correlation structures. Compared with the Log method, the proposed method shows larger correlations among the biomedical words (e.g., ``censor'', ``missing'', ``hazard''), a closer relationship between ``high'' and ``sparse'' and a less correlation between ``inference'' and ``missing''. The biplots for the Raw and Power methods provide little information on biostatistics, and moreover, some high loadings, for instance, ``test'' in Raw and ``regression'' in Power, lead others to decline.

\section{Discussion}\label{sec:discuss}
The paper connects the principal subspace of the compositional covariance matrix to that of the basis covariance matrix. The conceptual bridge relies on the simple relationship between the centered log-ratio transformed compositions and the basis. Remarkably, other transformations also have been proposed for principal component analysis for compositional data, for instance, the power transformation \citep{scealy2015robust}. It would be worthwhile to explore their connections with the principal subspace of the basis covariance as well. Notwithstanding, our approach has shown its great advantage of the direct interpretation for the principal subspace estimated from the compositions. 

Because of the unavailability of the basis, we take the centered log-ratio covariance matrix as a proxy, and thus, an identifiability issue emerges. To address this issue, we work under some subspace sparsity assumption and prove that the principal subspace of the basis covariance matrix is approximately identifiable with a diverging dimensionality.
The surprising blessing-of-dimensionality phenomenon is specific for compositional data. 
The sparsity assumption on the principal subspace is also necessary for the estimation in the high-dimensional setting, since standard principal component analysis is inconsistent owing to the curse of dimensionality. In fact, there are kinds of sparsity assumptions on principal component analysis for high-dimensional Euclidean data \citep{johnstone2009consistency,ma2013sparse,cai2013sparse,vu2013minimax}. We choose the notions of subspace sparsity in \cite{vu2013minimax} in this paper, while other definitions of sparsity also seem adoptable.
%since the analysis is conducted without the spiked covariance assumptions, which is more general in practice. 

%With some mild conditions, we establish the upper bounds for the principal subspace estimation, which indicates the dimension $p$ should be chosen judiciously to balance the approximation and estimation. When $p$ is large enough satisfying $n^{1-q/2}R_q^{q/(2-q)}\lesssim p(\log p)^{1-q/2}$, the approximation error could be bounded by the estimation error and we achieve the optimal rates. 

In practice, the relative abundances are usually observed directly, and compositions are estimated by the normalization. 
%If the observed relative abundances contain zero counts, we need to tackle the zero entries first. A simple way, for example, is replacing them with a small positive number. 
%Although these preprocessing procedures would cause an additional small error, the error bound of the principal subspace estimation could also be controlled.
For each subject $i$, assume that its relative abundances $(N_{i1},\ldots,N_{ip})^\T$ follows a multinomial distribution with the total number of counts $N_i=\sum_{j=1}^pN_{ij}$ and compositions $\X_i$.
The estimator of compositions is $\widehat\X_i=(N_{i1},\ldots,N_{ip})^\T/N_i$ and we have $\widehat\X_i = \X_i+\varrho_i$ with $\bE(\varrho_i|\X_i)=0$ and $\var(\varrho_i|\X_i) = (\diag(\X_i)-\X_i\X_i^\T)/N_i$.
We apply the centered log-ratio transformation to $\widehat\X_i$ and obtain $\Q_i=\clr(\widehat\X_i)=\Z_i+\bepsilon_i$, where $\bepsilon_i=\clr(1_p+\varrho_i/\X_i)$ with $/$ denoting the element-wise division.
It is easy to see that each element in $\bE(\bepsilon_i)$, $\var(\bepsilon_i)$ and $\cov(\Z_i,\bepsilon_i)$ is $O(1/N_i)$ if $X_i>c_1\ii_p$ for some constant $c_1>0$. 
Apply the proposed algorithm for the row sparsity to the sample covariance matrix of $\Q_i$ and denote the principal subspace estimation by $\widehat\cS_Q$. If $\bepsilon_i~(i=1,\ldots,n)$ are sub-Gaussian random vectors and $N_i=O(N)$, then the error bound is
\begin{equation*}
\|\sin\bTheta(\widehat{\cS}_Q,\cS_\bOmega)\|_F^2\leq c\left\{R_q\left(\frac{\sigma_1^2(d+\log p)}{n}\right)^{1-q/2}\vee\left(\frac{c(q)^2\sigma_2^2d^2R_q^{2/(2-q)}}{p}\right)\vee \frac{d^{1/2}p}{N}\right\}.
\end{equation*}
The additional error term is bounded by the first two terms if the total numbers of counts are sufficiently large, and in this case our result still holds.
%The preprocessing procedure would bring an addition small error term. 

\bibliographystyle{apalike}
\bibliography{pca}

\newpage
\appendix

\section{Proof of theorems}
\subsection{Davis-Kahan $\sin\theta$ theorem}
We will apply the Davis-Kahan $\sin\theta$ theorem in proofs of Theorems \ref{th:identify} and \ref{th:pca-Sz}.
\begin{theorem}
(\cite{davis1970rotation}). Let $\bSigma, \widehat{\bSigma} \in \bR^{p\times p}$ be symmetric, with eigenvalues $\lambda_1\geq\ldots\geq\lambda_p$ and $\widehat\lambda_1\geq\ldots\geq\widehat\lambda_p$ respectively. Fix $1\leq r\leq s\leq p$, let $d=s-r+1$, and let $\V=(\bfv_r,\bfv_{r+1},\ldots,\bfv_s)\in\bR^{p\times d}$ and $\widehat \V=(\widehat \bfv_r,\widehat \bfv_{r+1},\ldots,\widehat \bfv_s)\in \bR^{p \times d}$ have orthonormal columns satisfying $\bSigma \bfv_j=\lambda_j\bfv_j$ and $\widehat\bSigma\widehat \bfv_j=\widehat\lambda_j\widehat \bfv_j$ for $j = r,r+1,\ldots,s$. Let $\cS_\bSigma=\col(\V)$ and $\cS_{\widehat\bSigma}=\col(\widehat \V)$. If $\delta=\inf\{|\widehat\lambda-\lambda|:\lambda\in[\lambda_s,\lambda_r],\widehat\lambda\in(-\infty,\widehat\lambda_{s+1}]\cup[\widehat\lambda_{r-b1},\infty)\}>0$, where $\widehat\lambda_0=-\infty$ and $\widehat\lambda_{p+1}=\infty$, then for every unitary-invariant norm, $\delta\|\sin\bTheta(\cS_\bSigma,\cS_{\widehat\bSigma})\|\leq\|(\widehat\bSigma-\bSigma)\V\|$.
\end{theorem}

\subsection{Proof of Theorem \ref{th:identify}}\label{pf:identify}
\begin{proof}
We apply the Davis-Kahan $\sin\theta$ theorem to prove Theorem \ref{th:identify}.

Let $\bSigma = \bOmega, \widehat\bSigma = \bGamma, r = 1$, and $s=d$ in the Davis-Kahan $\sin\theta$ theorem, then for every unitary-invariant norm $\|\cdot\|$, we have 
\[
\|\sin\bTheta(\cS_\bOmega,\cS_\bGamma)\|\leq\frac{\|\bfH\V_0\|}{\delta},
\]
where $\bfH= \bGamma -\bOmega = -p^{-1}\bOmega\J-p^{-1}\J\bOmega+p^{-2}\J\bOmega\J$ and $\delta=\max(\lambda_d-a_{d+1},0)$.

Since
\[
\bfH\V_0 = -p^{-1}\bOmega\J \V_0-p^{-1}\J \V_0\bLambda_0+p^{-2}\J\bOmega\J \V_0,
\]
we have
\begin{align}\label{hv0}
\|\bfH\V_0\|_2& \leq p^{-1}\|\bOmega\|_2\|\J \V_0\|_2+p^{-1}\|\J \V_0\|_2\|\bLambda_0\|_2+p^{-2}\|\J\|_2\|\bOmega\|_2\|\J \V_0\|_2 \\ \notag
& \leq 3c(q)\lambda_1R_q^{1/(2-q)}\sqrt{(d/p)},
\end{align}
where Lemma \ref{lem:vjbound} is applied to bound $\|\J \V_0\|_2$
and
\[
\|\bfH\V_0\|_F\leq\sqrt{d}\|\bfH\V_0\|_2\leq 3c(q)\lambda_1R_q^{1/(2-q)}d/\sqrt{p}.
\]

For $\delta$ in the denominator, we claim that $\delta\geq\lambda_d-\lambda_{d+1}>0$, since $a_{d+1}\leq \lambda_{d+1}$ by Lemma \ref{lem:interlacing}.

Hence, we obtain
\[
\|\sin\bTheta(\cS_\bOmega,\cS_\bGamma)\|_2 \leq\frac{\|\bfH\V_0\|_2}{\delta} \leq \frac{3c(q)\lambda_1R_q^{1/(2-q)}\sqrt{(d/p)}}{\lambda_d-\lambda_{d+1}},
\]
and 
\[
\|\sin\bTheta(\cS_\bOmega,\cS_\bGamma)\|_F \leq\frac{\|\bfH\V_0\|_F}{\delta} \leq \frac{3c(q)\lambda_1R_q^{1/(2-q)}d/\sqrt{p}}{\lambda_d-\lambda_{d+1}}.
\]
\end{proof}

\ignore{\subsection{Proof of Theorem \ref{th:pca-Sy}}\label{pf:pca-Sy}
\begin{proof}
Adopting Davis-Kahan $\sin\theta$ theorem, we have
\begin{equation*}
\|\sin\bTheta(\widehat\cS_Y,\cS_\bOmega)\|_F\leq\frac{\|\bDelta \V_0\|_F}{\delta}\leq\frac{\sqrt{d}\|\bDelta\|_2}{\delta},
\end{equation*}
where $\bDelta =\bfS_Y-\bOmega$ and $\delta=\max(\lambda_d-\widehat{\lambda}_{d+1},0)$. 

Lemma \ref{lem:ineq-Sn} implies 
\begin{equation}\label{eq:pca-Sy}
\bP(\|\bDelta\|_2\geq t) \leq 2\cdot 7^p\exp\left(-\frac{t^2/18}{18\lambda_1^2/n+t\lambda_1/n}\right).
\end{equation}
Define 
\[
A_1 = \{\|\bDelta\|_2\geq t_1\}, \quad t_1 = 18\lambda_1\max\left(\sqrt{\frac{p+\log(1/\tau)}{n}},\frac{p+\log(1/\tau)}{n}\right).
\]
Taking $t=t_1$ into (\ref{eq:pca-Sy}), we have 
\[
\bP(A_1) = \bP(\|\bDelta\|_2\geq t_1) \leq \tau.
\]
Define
\[
A_\delta = \{\widehat{\lambda}_{d+1}\geq \lambda_{d+1}+t_\delta\}, \quad t_\delta = t_1.
\]
Since
\[
\widehat{\lambda}_{d+1} = \lambda_{d+1}(\bfS_Y)\leq \lambda_{d+1}+\|\bDelta\|_2
\]
by Lemma \ref{lem:weyl}, we obtain
\[
\bP(A_\delta) \leq \bP(\|\bDelta\|_2\geq t_\delta) \leq \tau
\]
Thus, on $A_1^c\cap A_\delta^c$,
\[
\|\sin\bTheta(\widehat\cS_Y,\cS_\bOmega)\|_F \leq c\lambda_1\sqrt{d} \max\left(\sqrt{\frac{p+\log(1/\tau)}{n}},\frac{p+\log(1/\tau)}{n}\right)/\delta,
\]
where $\delta = \max(\lambda_d-\lambda_{d+1}-t_\delta,0)$ with probability at least $1-2\tau$.
\end{proof}}

\subsection{Proof of Theorem \ref{th:pca-Sz}}\label{pf:pca-Sz}
\begin{proof}
Adopting the Davis-Kahan $\sin\theta$ theorem, we have
\begin{equation}\label{pca-0}
\|\sin\bTheta(\widehat\cS_Z,\cS_\bOmega)\|_F\leq\frac{\|\bDelta \V_0\|_F}{\delta}\leq\frac{\sqrt{d}\|\bDelta \V_0\|_2}{\delta},
\end{equation}
where $\bDelta =\bfS_Z-\bOmega$ and $\delta=\max(\lambda_d-\widehat{\lambda}_{d+1},0)$.

Since
\[
\bDelta = \bfS_Z-\bGamma+\bfH,
\]
where $\bfH = \bGamma-\bOmega$,
we have
\begin{align*}
\|\bDelta \V_0\|_2 \leq &~ \|\bfS_Z-\bGamma\|_2 + \|\bfH\V_0\|_2 \\
= &~ T_1 + T_2. 
\end{align*}

We first bound the term $T_1$. Lemma \ref{lem:ineq-Sn} implies 
\begin{align}\label{eq:pca-Sz}
& ~\bP(\|\bfS_Z-\bGamma\|_2\geq t) \leq 2\cdot 7^p\exp\left(-\frac{t^2/18}{18a_1^2/n+ta_1/n}\right) \leq 2\cdot 7^p\exp\left(-\frac{t^2/18}{18\lambda_1^2/n+t\lambda_1/n}\right),
\end{align}
where $a_1\leq \lambda_1$ by Lemma \ref{lem:interlacing} gives the last inequality.
Define 
\[
A_1 = \{T_1\geq t_1\}, \quad t_1 = 18\lambda_1\max\left(\left(\frac{p-\log\tau}{n}\right)^{1/2},\frac{p-\log\tau}{n}\right).
\]

Plugging $t=t_1$ into (\ref{eq:pca-Sz}), we have 
\[
\bP(A_1) = \bP(\|\bfS_Z-\bGamma\|_2\geq t_1) \leq \tau.
\]

To bound the term $T_2$, we use (\ref{hv0}) and obtain
\[
T_2 = \|\bfH\V_0\|_2 \leq 3c(q)\lambda_1R_q^{1/(2-q)}\sqrt{(d/p)}.
\]

Now we consider $\delta=\max(\lambda_d-\widehat{\lambda}_{d+1},0)$.
By Lemma \ref{lem:weyl}, we have
\[
\widehat{\lambda}_{d+1}=\lambda_{d+1}(\bfS_Z)\leq a_{d+1}+\|\bfS_Z-\bGamma\|_2.
\]
Since $a_{d+1}\leq\lambda_{d+1}$ by Lemma \ref{lem:interlacing},
we obtain
\[
\widehat{\lambda}_{d+1} \leq \lambda_{d+1} + \|\bfS_Z-\bGamma\|_2.
\]
Define
\[
A_\delta =\{\widehat{\lambda}_{d+1} \geq \lambda_{d+1} + t_\delta\}, \quad t_\delta = t_1.
\]
and we have
\[
\bP(A_\delta) \leq \bP(\|\bfS_Z-\bGamma\|_2\geq t_\delta)\leq \tau.
\]

Now on $A_1^c\cap A_\delta^c$, we have
\begin{align}\label{pca-1}
\|\bDelta \V_0\|_2\leq 72\lambda_1\max\left(\left(\frac{p-\log\tau}{n}\right)^{1/2},\frac{p-\log\tau}{n}\right) +3c(q)\lambda_1R_q^{1/(2-q)}\sqrt{(d/p)}
\end{align}
and 
\begin{equation}\label{pca-2}
\lambda_d-\widehat{\lambda}_{d+1}\geq \lambda_d-\lambda_{d+1}-t_\delta.
\end{equation}
Plugging (\ref{pca-1}) and (\ref{pca-2}) into (\ref{pca-0}), we have
\begin{align*}
&\|\sin\bTheta(\widehat{\cS}_Z,\cS_\bOmega)\|_F \\
\leq &c\lambda_1\left\{\sqrt{d}\max\left(\left(\frac{p-\log\tau}{n}\right)^{1/2},\frac{p-\log\tau}{n}\right)\vee\frac{c(q)R_q^{1/(2-q)}d}{\sqrt{p}}\right\}/{\delta},
\end{align*}
where $\delta=\max\left(\lambda_d-\lambda_{d+1}-t_\delta,0\right)$
with probability at least $1-2\tau$.
\end{proof}

\subsection{Proof of Theorem \ref{th:upperbound1}}\label{pf:upperbound}
\begin{proof}
We will use some techniques as in \cite{vu2013minimax}. We start from applying Lemma \ref{lem0}, which gives
\begin{equation}\label{eq:error}
\widehat\v^2=\|\sin\bTheta(\widehat\cS_Z,\cS_\bOmega)\|_F^2\leq\frac{\langle \bfS_Z-\bOmega, \widehat{\V}_0\widehat{\V}_0^\T-\V_0\V_0^\T\rangle}{\lambda_d-\lambda_{d+1}}.
\end{equation}
Let
\[
\W_\bGamma = \bfS_Z-\bGamma \quad\text{ and }\quad \bfH= \bGamma - \bOmega. 
\]
Let
\[
\bPi = \V_0\V_0^\T \quad\text{ and } \quad\widehat\bPi = \widehat{\V}_0\widehat{\V}_0^\T.
\]
For an orthogonal projector $\bPi$, we write $\bPi^\bot=\I-\bPi=\V_1\V_1^\T$. The numerator in (\ref{eq:error}) can be decoupled into estimation error part and approximation error part:
\[
\langle \bfS_Z-\bOmega, \widehat{\V}_0\widehat{\V}_0^\T-\V_0\V_0^\T\rangle = \langle \W_\bGamma,\widehat\bPi-\bPi \rangle + \langle \bfH,\widehat\bPi-\bPi \rangle,
\] 
where
\begin{align*}
\langle \W_\bGamma,\widehat\bPi-\bPi \rangle & = -\langle \W_\bGamma,\bPi\widehat\bPi^\bot\bPi \rangle + 2\langle \W_\bGamma,\bPi^\bot\widehat\bPi\bPi \rangle + \langle \W_\bGamma,\bPi^\bot\widehat\bPi\bPi^\bot \rangle \\ \notag
& = -T_1 + 2T_2 + T_3
\end{align*}
by Proposition \ref{prop:eq}, and
\[
\langle \bfH,\widehat{\bPi}-\bPi\rangle = \langle \bfH,\widehat{\bPi}\rangle - \langle \bfH,\bPi\rangle=T_4 - T_5.
\]
We will control $T_1, \ldots, T_5$ separately.

For the term $T_1$, 
\begin{align*}
|T_1| & = |\langle \W_\bGamma,\bPi\widehat{\bPi}^\bot\bPi\rangle| = |\langle\bPi \W_\bGamma\bPi,\bPi\widehat{\bPi}^\bot\bPi\rangle| \\
& \leq \|\bPi \W_\bGamma\bPi\|_2\|\bPi\widehat{\bPi}^\bot\bPi\|_* = \|\bPi \W_\bGamma\bPi\|_2\|\bPi\widehat{\bPi}^\bot\widehat{\bPi}^\bot\bPi\|_* \\
& = \|\bPi \W_\bGamma\bPi\|_2\|\bPi\widehat{\bPi}^\bot\|_F^2 = \|\bPi \W_\bGamma\bPi\|_2\widehat{\v}^2.
\end{align*}
%Then applying Lemma \ref{lem:T1}, we have
%\[
%\|\|\bPi \W_\bGamma\bPi\|_2\|_{\psi_1} \leq c_1\lambda_1\left(1+c(q)^2R_q^{\frac{2}{2-q}}d/p\right)\sqrt{d/n},
%\]
%where $c_1$ is a universal constant. Define
%\[
%A_1 = \left\{|T_1|\geq c_1\lambda_1\log n\left(1+c(q)^2R_q^{\frac{2}{2-q}}d/p\right)\sqrt{d/n}\widehat{\v}^2\right\}.
%\]
%Then, when $n\geq 2$ we have
%\[
%\bP(A_1)\leq\bP\left(\|\bPi \W_\bGamma\bPi\|_2\geq c_1\lambda_1\log n\left(1+c(q)^2R_q^{\frac{2}{2-q}}d/p\right)\sqrt{d/n}\right)\leq (n-1)^{-1}.
%\]
Lemma \ref{lem:T1} implies
\begin{equation}\label{eq:T1}
\bP(\|\bPi \W_\bGamma\bPi\|_2\geq t) \leq 2\cdot 7^d \exp\left(-\frac{t^2/18}{18\alpha^2/n+\alpha t/n}\right),
\end{equation}
where $\alpha = \lambda_1\left(1+2c(q)^2R_q^{2/(2-q)}d/p\right)$.
Define
\[
A_1 :\left\{|T_1| \geq t_1\widehat{\v}^2\right\}, \quad t_1 = c_1\alpha\left(\frac{d+\log n}{n}\right)^{1/2},
\]
where $c_1$ is a large enough positive constant. Plugging $t = t_1$ into (\ref{eq:T1}), we have
\begin{align*}
\bP(A_1) & \leq \bP(\|\bPi \W_\bGamma\bPi\|_2\geq t_1)\leq n^{-1}.
\end{align*}

For the term $T_2$, 
\begin{align}\label{eq:T2-1}
|T_2| & = |\langle \W_\bGamma,\bPi^\bot\widehat\bPi\bPi \rangle| = |\langle \bPi^\bot \W_\bGamma\bPi,\bPi^\bot\widehat\bPi \rangle| \\ \notag
& \leq \|\bPi^\bot \W_\bGamma\bPi\|_{2,\infty}\|\bPi^\bot\widehat\bPi\|_{2,1}.
\end{align}
To bound $\|\bPi^\bot\widehat\bPi\|_{2,1}$, let the rows of $\bPi^\bot\widehat\bPi$ be denoted by $\bphi_1,\ldots,\bphi_p$ and $t>0$. Using a standard argument of bounding $l_1$ norm by the $l_q$ and $l_2$ norms [for example, from Lemma 5 of \cite{raskutti2011minimax}], we have for all $t>0$, $0<q\leq 1$, 
\begin{align}\label{eq:T2-2}
\|\bPi^\bot\widehat\bPi\|_{2,1} = &~\sum_{i=1}^p\|\bphi_i\|_2 \\ \notag
\leq &~\left(\sum_{i=1}^p\|\bphi_i\|_2^q\right)^{1/2} \left(\sum_{i=1}^p\|\bphi_i\|_2^2\right)^{1/2}t^{-q/2} + \left(\sum_{i=1}^p\|\bphi_i\|_2^q\right)t^{1-q} \\ \notag
= &~ \|\bPi^\bot\widehat\bPi\|_{2,q}^{q/2}\|\bPi^\bot\widehat\bPi\|_F t^{-q/2} + \|\bPi^\bot\widehat\bPi\|_{2,q}^qt^{1-q} \\ \notag
\leq &~\sqrt{(2R_q)}t^{-q/2}\widehat\v + 2R_qt^{1-q},
\end{align}
where the last step uses the fact that
\begin{align*}
\|\bPi^\bot\widehat\bPi\|_{2,q}^q & = \|\bPi^\bot\widehat \V_0\|_{2,q}^q = \|\widehat \V_0-\bPi\widehat \V_0\|_{2,q}^q \leq \|\widehat \V_0\|_{2,q}^q + \|\V_0\V_0^\T\widehat \V_0\|_{2,q}^q \\
& \leq \|\widehat \V_0\|_{2,q}^q + \|\V_0\|_{2,q}^q \leq 2R_q.
\end{align*}
When $q=0$, for all $t>0$, we have
\begin{equation}\label{eq:T2-2-0}
\|\bPi^\bot\widehat\bPi\|_{2,1}\leq\|\bPi^\bot\widehat\bPi\|_{2,0}^{1/2}\|\bPi^\bot\widehat\bPi\|_F+\|\bPi^\bot\widehat\bPi\|_{2,0}t\leq\sqrt{(2R_0)}\widehat\v+2R_0t,
\end{equation}
where the last step uses the fact that
\begin{align*}
\|\bPi^\bot\widehat\bPi\|_{2,0}& = \|\bPi^\bot\widehat \V_0\|_{2,0} = \|\widehat \V_0-\bPi\widehat \V_0\|_{2,0} \leq \|\widehat \V_0\|_{2,0} + \|\V_0\V_0^\T\widehat \V_0\|_{2,0} \\
& \leq \|\widehat \V_0\|_{2,0} + \|\V_0\|_{2,0} \leq 2R_0.
\end{align*}
Combining (\ref{eq:T2-1}), (\ref{eq:T2-2}) and (\ref{eq:T2-2-0}) we obtain, for all $t>0$, $0\leq q\leq 1$,
\begin{equation}\label{eq:T2-3}
|T_2|\leq\|\bPi^\bot \W_\bGamma\bPi\|_{2,\infty}(2^{1/2}R_q^{1/2}t^{-q/2}\widehat\v+2R_qt^{1-q}).
\end{equation}
Now define
\begin{align*}
A_2 = & \left\{|T_2|\geq t_{2,1}\left(2^{1/2}R_q^{1/2}t_{2,2}^{-q/2}\widehat\v + 2R_qt_{2,2}^{1-q}\right)\right\}, \\
t_{2,1} = & 20\beta\left(\frac{d+\log p}{n}\right)^{1/2}, \\
t_{2,2} = & \frac{\sqrt{(\lambda_1\lambda_{d+1})}}{\lambda_d-\lambda_{d+1}}\left(\frac{d+\log p}{n}\right)^{1/2},
\end{align*}
where $\beta = \sqrt{(\lambda_1\lambda_{d+1})} + 7c(q)\lambda_1R_q^{1/(2-q)}\sqrt{(d/p)}$.
Taking $t=t_{2,2}$ in (\ref{eq:T2-3}) and using the tail bound result in Lemma \ref{lem:T2}, we have
\begin{align}
\bP(A_2) & \leq \bP(\|\bPi^\bot \W_\bGamma\bPi\|_{2,\infty}\geq t_{2,1}) \\ \notag
& \leq 2p5^d\exp\left(-\frac{t_{2,1}^2/8}{2\beta^2/n+t_{2,1}\beta/n}\right) \\ \notag
& \leq p^{-1}.
\end{align}

For the term $T_3 = \langle \W_\bGamma,\bPi^\bot\widehat\bPi\bPi^\bot \rangle = \langle \bfS_Z-\bGamma,\bPi^\bot\widehat\bPi\bPi^\bot \rangle $,
we use the same bound in \cite{vu2013minimax}. Define
\[
A_3 = \left\{|T_3|\geq c_3(\log n)^{5/2}\lambda_{d+1}\left(\v_n\widehat\v^2 + \v_n^2\widehat\v + \v_n^4\right)\right\}
\]
and we have 
\[
\bP(A_3) \leq \frac{6\log n}{n} + \frac{3}{n}.
\]

For the term $T_4$, recalling that
\[
\bfH = -p^{-1}\bOmega\J-p^{-1}\J\bOmega + p^{-2}\J\bOmega\J,
\]
we have
\begin{align}\label{eq:T4}
|T_4| & = |\langle \bfH,\widehat{\bPi}\rangle| = |\langle \widehat{\bPi},\bfH\widehat{\bPi}\rangle| \\ \notag
& \leq p^{-1}|\langle \widehat{\bPi},\bOmega\J\widehat{\bPi}\rangle| + p^{-1}|\langle \widehat{\bPi},\J\bOmega\widehat{\bPi}\rangle| + p^{-2}|\langle \widehat{\bPi},\J\bOmega\J\widehat{\bPi}\rangle| \\ \notag
& = 2p^{-1}|\langle \widehat{\bPi},\bOmega\J\widehat{\bPi}\rangle| + p^{-2}|\langle \widehat{\bPi},\J\bOmega\J\widehat{\bPi}\rangle| \\ \notag
& = 2p^{-1}T_{4,1} + p^{-2}T_{4,2}.
\end{align}
To control $T_{4,1}$, we use
\[
\langle\widehat{\bPi},\bOmega\J\widehat{\bPi}\rangle = \langle \bLambda_0\V_0^\T\widehat{\bPi},\V_0^\T\J\widehat{\bPi}\rangle + \langle \V_1\widehat{\bPi},\bLambda_1 \V_1^\T\J\widehat{\bPi}\rangle.
\]
Since
\[
\|\widehat{\bPi}\V_0\bLambda_0\|_*\leq d \|\widehat{\bPi}\V_0\bLambda_0\|_2\leq \lambda_1d,
\]
and
\[
\|\V_0^\T\J\widehat{\bPi}\|_2\leq\|\V_0^\T\ii\ii^\T\widehat{\V}_0\|_2\leq\|\ii^\T\widehat{\V}_0\|_2^2\leq d\|\ii^\T\widehat{\V}_0\|_{\max}^2\leq c(q)^2dR_q^{2/(2-q)},
\]
where Lemma \ref{lem:vjbound} is adopted in the last step,
we have
\[
|\langle \bLambda_0\V_0^\T\widehat{\bPi},\V_0^\T\J\widehat{\bPi}\rangle|\leq \|\widehat{\bPi}\V_0\bLambda_0\|_*\|\V_0^\T\J\widehat{\bPi}\|_2\leq c(q)^2\lambda_1d^2R_q^{2/(2-q)}.
\]
Since
\[
\|\V_1\widehat{\bPi}\|_*\leq \sqrt{d}\|\widehat{\bPi}\V_1\|_F=\sqrt{d}\widehat{\v},
\] and
\[
\|\bLambda_1 \V_1^\T\J\widehat{\bPi}\|_2\leq \lambda_{d+1}\|\J\widehat{\V}_0\|_2\leq c(q)\lambda_{d+1}\sqrt{(pd)}R_q^{1/(2-q)},
\]
where Lemma \ref{lem:vjbound} is adopted in the last step,
we obtain
\[
|\langle \V_1\widehat{\bPi},\bLambda_1 \V_1^\T\J\widehat{\bPi}\rangle|\leq \|\V_1\widehat{\bPi}\|_*\|\bLambda_1 \V_1^\T\J\widehat{\bPi}\|_2 \leq c(q)\lambda_{d+1}d\sqrt{p}R_q^{1/(2-q)}\widehat{\v}.
\]
Thus, 
\begin{equation}\label{eq:T4-1}
|T_{4,1}| = |\langle \widehat{\bPi},\bOmega\J\widehat{\bPi}\rangle|\leq c(q)^2\lambda_1d^2R_q^{2/(2-q)} + c(q)\lambda_{d+1}d\sqrt{p}R_q^{1/(2-q)}\widehat{\v}.
\end{equation}
To control $T_{4,2}$, we use
\[
|T_{4,2}| = |\langle\widehat{\bPi},\J\bOmega\J\widehat{\bPi}\rangle| \leq \|\widehat{\bPi}\J\|_*\|\bOmega\J\widehat{\bPi}\|_2. 
\]
Since
\[
\|\widehat{\bPi}\J\|_*\leq \sqrt{d}\|\widehat{\bPi}\J\|_F=\sqrt{d}\|\J\widehat{\V_0}\|_F\leq c(q)d\sqrt{p}R_q^{1/(2-q)},
\]
where Lemma \ref{lem:vjbound} is adopted in the last step, and
\[
\|\bOmega\J\widehat{\bPi}\|_2\leq \lambda_1\|\J\widehat{\V}_0\|_2\leq c(q)\lambda_1\sqrt{(pd)}R_q^{1/(2-q)},
\]
\begin{equation}\label{eq:T4-2}
|T_{4,2}|\leq c(q)^2\lambda_1d\sqrt{d}pR_q^{2/(2-q)}.
\end{equation}
Plugging (\ref{eq:T4-1}) and (\ref{eq:T4-2}) into (\ref{eq:T4}), we obtain
\[
|T_4|\leq 3c(q)^2\lambda_1R_q^{2/(2-q)}d^2/p + 2c(q)\lambda_{d+1}R_q^{1/(2-q)}\widehat{\v}d/\sqrt{p}.
\]

For the term $T_5=\langle \bfH,\bPi\rangle$, we have
\[
|T_5|=|\langle \bfH,\bPi\rangle| = |\langle \bPi,\bPi \bfH\bPi\rangle| \leq \|\bPi\|_*\|\bPi \bfH\bPi\|_2,
\]
where
 $\|\bPi\|_*\leq d\|\bPi\|_2=d$,
and
$\|\bPi \bfH\bPi\|_2\leq \|\V_0^\T\bfH\V_0\|_2\leq 3c(q)^2\lambda_1R_q^{2/(2-q)}d/p$ by Lemma \ref{lem:v0hv0}.
Hence,
\[
|T_5|\leq 3c(q)^2\lambda_1R_q^{2/(2-q)}d^2/p.
\]

On $A_1^c\cap A_2^c\cap A_3^c$, plugging all terms into (\ref{eq:error}) and under Conditions \ref{condition1} and \ref{condition2}, we have that
\begin{align*}
(\lambda_d-\lambda_{d+1})\widehat\v^2 \leq & ~\left\{c_1\lambda_1\left(1+2c(q)^2R_q^{2/(2-q)}d/p\right)(d+\log n)^{1/2}/n^{1/2}+ c_3\lambda_{d+1}(\log n)^{5/2}\v_n\right\}\widehat\v^2 \\
& + 40\left\{(\lambda_1\lambda_{d+1})^{1/2}+7c(q)\lambda_1R_q^{1/(2-q)}\sqrt{(d/p)}\right\}(\lambda_1\lambda_{d+1})^{-q/4}(\lambda_d-\lambda_{d+1})^{q/2}\v_n\widehat\v \\
& + 40\left\{(\lambda_1\lambda_{d+1})^{1/2}+7c(q)\lambda_1R_q^{1/(2-q)}\sqrt{(d/p)}\right\}(\lambda_1\lambda_{d+1})^{(1-q)/2}(\lambda_d-\lambda_{d+1})^{q-1}\v_n^2 \\
& + 2c(q)\lambda_{d+1}R_q^{1/(2-q)}d/\sqrt{p}\widehat\v + 6c(q)^2\lambda_1R_q^{2/(2-q)}d^2/p\\
& + c_3\lambda_{d+1}(\log n)^{5/2}\v_n^2(\widehat\v + \v_n^2).
\end{align*}
Therefore,
\begin{align*}
\frac{1}{2}(\lambda_d-\lambda_{d+1})\widehat\v^2 \leq & \frac{k_1}{2}(\lambda_1\lambda_{d+1})^{1/2-q/4}(\lambda_d-\lambda_{d+1})^{q/2}\v_n\widehat\v + \frac{k_1}{2}(\lambda_1\lambda_{d+1})^{1-q/2}(\lambda_d-\lambda_{d+1})^{q-1}\v_n^2 \\
& + \frac{k_2}{2}c(q)\lambda_{d+1}R_q^{1/(2-q)}d/\sqrt{p}\widehat\v + \frac{k_2}{2}c(q)^2\lambda_1R_q^{2/(2-q)}d^2/p,
\end{align*}
where $k_1 = 82+560c_2$ and $k_2 = 12$.
Let 
\[
\cE_1 = \sigma_1^{1-q/2}\v_n \quad\text{ and } \quad \cE_2 = c(q)\sigma_2R_q^{1/(2-q)}d/\sqrt{p},
\]
where $\sigma_1^2 = \lambda_1\lambda_{d+1}/(\lambda_d-\lambda_{d+1})^2$, $\sigma_2^2 = \lambda_1^2/(\lambda_d-\lambda_{d+1})^2$ and $\v_n = \sqrt{2R_q}(d+\log p/n)^{1/2-q/4}$.
We obtain 
\[
\widehat{\v}^2 - (k_1\cE_1+k_2\cE_2)\widehat{\v} - (k_1\cE_1^2+k_2\cE_2^2) \leq 0.
\]
Thus,
\begin{align*}
\widehat\v \leq  (k_1+1)\cE_1+(k_2+1)\cE_2 \leq c\left\{\sqrt{R_q}\sigma_1^{1-q/2}\left(\frac{d+\log p}{n}\right)^{1/2-q/4}\vee\left(c(q)\sigma_2\frac{dR_q^{1/(2-q)}}{\sqrt{p}}\right)\right\}
\end{align*}
with probability at least $1-4/n-6\log n/n-1/p$.
\end{proof}

\section{Related lemmas and propositions}
\begin{lemma}\label{lem:ineq-Sn}
Let $\X_1,\ldots,\X_n\in\bR^p$ be $n$ i.i.d. sub-Gaussian random vectors such that $\bE(\X_i\X_i^\T)=\bSigma$ and $\bE(\X_i)=0$. The sample covariance matrix $\widehat{\bSigma}$ defined by
\[
\widehat{\bSigma} = \frac{1}{n}\sum_{i=1}^n\X_i\X_i^\T.
\]
Then 
\begin{align*}
\bP(\|\widehat{\bSigma}-\bSigma\|_2\geq t) \leq 2\cdot 7^p\exp\left(-\frac{t^2/18}{18\|\bSigma\|_2^2/n+t\|\bSigma\|_2/n}\right)
\end{align*}
and
\begin{align*}
\bP(\|\widehat{\bSigma}-\bSigma\|_{\max}\geq t) \leq 2p^2\exp \left(-\frac{t^2/2}{18\|\bSigma\|_2^2/n+3t\|\bSigma\|_2/n}\right).
\end{align*}
\end{lemma}
\begin{proof}

Let $\cN_\delta$ be a minimal $\delta$-net of $S_2^{p-1}$ for some $\delta\in(0,1)$. Proposition \ref{prop:net2} implies
\[
\|\widehat{\bSigma}-\bSigma\|_2 \leq (1-2\delta)^{-1}\max_{\bfu\in\cN_\delta}|\langle \bfu,(\widehat{\bSigma}-\bSigma)\bfu\rangle|.
\]
Note that
\[
\langle \bfu,(\widehat{\bSigma}-\bSigma)\bfu\rangle = \frac{1}{n}\sum_{i=1}^n\left(( \X_i^\T\bfu)^2- \langle \bfu,\bSigma \bfu\rangle\right),
\]
and 
\begin{align*}
\|(\X_i^\T\bfu)^2 - \langle \bfu,\bSigma \bfu\rangle\|_{\psi_1} & \leq \|(\X_i^\T\bfu)^2 \|_{\psi_1} + 2\bfu^\T\bSigma \bfu \\
& \leq \|\X_i^\T\bfu\|_{\psi_2}^2 + 2\|\bSigma\|_2 \\
& \leq 3\|\bSigma\|_2,
\end{align*}
where Proposition \ref{prop:ineq} is applied in the second inequality. 

We choose 
$\delta = 1/3$ and have $|\cN_\delta|\leq 7^p$ by Proposition \ref{prop:netnum}. 
Then Lemma \ref{lem:bern} gives
\begin{align*}
\bP(\|\widehat{\bSigma}-\bSigma\|_2\geq t) & \leq \bP(\max_{\bfu\in\cN_\delta}|\langle \bfu,(\widehat{\bSigma}-\bSigma)\bfu\rangle|\geq t/3) \\
& \leq |N_\delta| \bP(|\langle \bfu,(\widehat{\bSigma}-\bSigma)\bfu\rangle|\geq t/3) \\
& \leq 2\cdot 7^p\exp\left(-\frac{t^2/18}{18\|\bSigma\|_2^2/n+t\|\bSigma\|_2/n}\right).
\end{align*}

Note that
\[
(\widehat{\bSigma}-\bSigma)_{jk} = \frac{1}{n}\sum_{i=1}^n(X_{ij}X_{ik}-\Sigma_{jk}),
\]
and 
\begin{align*}
\|X_{ij}X_{ik}-\Sigma_{jk}\|_{\psi_1} & \leq \|X_{ij}X_{ik}\|_{\psi_1} + 2|\Sigma_{jk}| \\
& \leq \|X_{ij}\|_{\psi_2} \|X_{ik}\|_{\psi_2} + 2\sqrt{(\Sigma_{jj}\Sigma_{kk})} \\
& \leq 3\|\bSigma\|_2,
\end{align*}
where Proposition \ref{prop:ineq}  is applied in the second inequality and $\Sigma_{jj}\leq \|\bSigma\|_2$ for $j=1,\ldots,p$ gives the last step.
Then Lemma \ref{lem:bern} implies 
\begin{align*}
\bP(|(\widehat{\bSigma}-\bSigma)_{jk}| \geq t) \leq 2\exp \left(-\frac{t^2/2}{18\|\bSigma\|_2^2/n+3t\|\bSigma\|_2/n}\right).
\end{align*}
Hence,
\begin{align*}
\bP(\|\widehat{\bSigma}-\bSigma\|_{\max}\geq t) &\leq p^2 \bP(|(\widehat{\bSigma}-\bSigma)_{jk}| \geq t) \\
& \leq 2p^2\exp \left(-\frac{t^2/2}{18\|\bSigma\|_2^2/n+3t\|\bSigma\|_2/n}\right).
\end{align*}

\end{proof}

\begin{lemma}\label{lem:vjbound}
If $q\in[0,1]$ and $\cS_\bOmega = \text{span}\{\bfv_1,\ldots,\bfv_d\} = \col(\V_0)\in \cM_q(R_q)\cup\cM_q^*(R_q)$, then 
\[
\max_{j=1,\ldots,d}\|\bfv_j\|_1 \leq c(q)R_q^{1/(2-q)}
\]
\[
\|\J \V_0\|_2\leq c(q)\sqrt{(pd)}R_q^{1/(2-q)}\quad \text{ and }\quad \|\J \V_0\|_F\leq c(q)\sqrt{(pd)}R_q^{1/(2-q)},
\]
where
\[
c(q)=\frac{2-q}{2(1-q)}\left\{\frac{2(1-q)}{q}\right\}^{q/(2-q)}\text{I}(q\in (0,1)) + 2\text{I}(q\in\{0,1\})
\]
and $\text{I}(\cdot)$ is an indicator function.
\end{lemma}
\begin{proof}
We first show that $\max_{j=1,\ldots,d}\|\bfv_j\|_q^q\leq R_q$ as follows.

For $q\in(0,1]$, we have
\[
R_q\geq \sum_{i=1}^p\|\bfv_{i*}\|_2^q\geq\sum_{i=1}^p|\bfv_{i*}|_{\max}^q\geq\max_{j=1,\ldots,d}\|\bfv_j\|_q^q
\]
when $\cS_\bOmega\in\cM_q(R_q)$ ($R_q\geq\max_{j=1,\ldots,d}\|\bfv_j\|_q^q$ when $\cS_\bOmega\in\cM^*_q(R_q)$), where $\bfv_{i*}$ denotes the $i$th row of $\V_0$.

For $q=0$, we have
\[
R_0\geq\sum_{i=1}^pI(\|\bfv_{i*}\|_2\neq 0)\geq\max_{j=1,\ldots,d}\|\bfv_j\|_0
\]
when $\cS_\bOmega\in\cM_q(R_q)$ ($R_q\geq\max_{j=1,\ldots,d}\|\bfv_j\|_0$ when $\cS_\bOmega\in\cM^*_q(R_q)$).

Then applying a standard argument of bounding $l_1$ norm by the $l_q$ and $l_2$ norms [for example, from Lemma 5 of \cite{raskutti2011minimax}] and noticing that $\|\bfv_j\|_2=1$ for $j=1,\ldots,d$, we have
\begin{equation}\label{vjbound}
\max_{j=1,\ldots,d}\|\bfv_j\|_1\leq \sqrt{R_q}\tau^{-q/2} + R_q\tau^{1-q}
\end{equation}
for all $\tau>0$ and $0\leq q\leq 1$.

For $q\in (0,1)$, let $\tau = (2\sqrt{R_q}(1-q)/q)^{2/(q-2)}$ in (\ref{vjbound}). We have 
\[
\max_{j=1,\ldots,d}\|\bfv_j\|_1\leq \frac{2-q}{2(1-q)}\left\{\frac{2(1-q)}{q}\right\}^{q/(2-q)}R_q^{1/(2-q)}.
\]

For $q=0$, plugging $\tau=1/\sqrt{R_q}$ into (\ref{vjbound}), we obtain
\[
\max_{j=1,\ldots,d}\|\bfv_j\|_1\leq 2\sqrt{R_q}.
\]

For $q=1$, plugging $\tau=1/R_q$ into (\ref{vjbound}), we obtain
\[
\max_{j=1,\ldots,d}\|\bfv_j\|_1\leq 2R_q.
\]

Hence,
\begin{align*}
\|\J \V_0\|_2 & \leq\sqrt{(pd)}\|\J \V_0\|_{\max} \\
& =\sqrt{(pd)}\max_{j=1,\ldots,d}|\ii^\T\bfv_j|\leq\sqrt{(pd)}\max_{j=1,\ldots,d}\|\bfv_j\|_1\leq c(q)\sqrt{(pd)}R_q^{1/(2-q)},\\
\|\J \V_0\|_F & \leq\sqrt{(pd)}\|\J \V_0\|_{\max}\leq c(q)\sqrt{(pd)}R_q^{1/(2-q)}.
\end{align*}
\end{proof}

\begin{lemma}\label{lem:v0hv0}
If $q\in[0,1]$, $\bfH=\bGamma-\bOmega$ and $\cS_\bOmega = \text{span}\{\bfv_1,\ldots,\bfv_d\} = \col(\V_0)\in \cM_q(R_q)\cup\cM_q^*(R_q)$, then 
\[
\|\V_0^\T \bfH \V_0\|_2\leq 3c(q)^2\lambda_1R_q^{2/(2-q)}d/p.
\]
\end{lemma}
\begin{proof}
Noticing that
\[
\bOmega = \V_0\bLambda_0\V_0^\T + \V_1\bLambda_1\V_1^\T \text{ and } \bfH = \bGamma-\bOmega = -p^{-1}\bOmega\J - p^{-1}\J\bOmega + p^{-2}\J\bOmega\J,
\]
we have
\[
\V_0^\T \bfH\V_0 = -p^{-1}\bLambda_0 \V_0^\T\J \V_0 - p^{-1}\V_0^\T\J \V_0\bLambda_0 + p^{-2}\V_0^\T\J\bOmega\J \V_0.
\]
The result of Lemma \ref{lem:vjbound} implies
\begin{align*}
\|\bLambda_0\V_0^\T\J \V_0\|_2=&~\|\V_0^\T\J \V_0\bLambda_0\|_2 \\
\leq & ~\lambda_1 \|\V_0^\T\ii\ii^\T\V_0\|_2 \leq \lambda_1 \|\ii^\T\V_0\|_2^2 \\
\leq & ~\lambda_1 d\|\ii^\T\V_0\|_{\max}^2 \leq \lambda_1 d\max_{j=1,\ldots,d}\|\bfv_j\|_1^2 \leq c(q)^2\lambda_1 dR_q^{2/(2-q)},
\end{align*}
and
\[
\|\V_0^\T\J\bOmega \J \V_0\|_2\leq \lambda_1\|\J \V_0\|_2^2 \leq c(q)^2\lambda_1 pdR_q^{2/(2-q)}.
\]
Thus,
\[
\|\V_0^\T\bfH\V_0\|_2 \leq 2p^{-1}\|\bLambda_0\V_0^\T\J \V_0\|_2 + p^{-2}\|\V_0^\T\J\bOmega\J \V_0\|_2\leq 3c(q)^2\lambda_1R_q^{2/(2-q)}d/p.
\]

\end{proof}

\begin{lemma}\label{lem:interlacing}
(Cauchy interlacing). Let $\A$ be a symmetric $n\times n$ matrix. The $m\times m$ matrix $\B$, where $m\leq n$, is called a compression of $\A$ if there exists an orthogonal projection $\bfP$ onto a subspace of dimension $m$ such that $\bfP^\T\A \bfP = \B$.

If the eigenvalues of $\A$ are $\alpha_1\geq\ldots\geq\alpha_n$, and those of $\B$ are $\beta_1\geq\ldots\geq\beta_m$, then for all $j\leq m$,
\[
\alpha_{n-m+j}\leq\beta_j\leq\alpha_j.
\]
Notice that, when $n=m+1$, we have $\alpha_{j+1}\leq\beta_j\leq\alpha_j$.
\end{lemma}

\begin{lemma}\label{lem:weyl}
(Weyl's inequality). Let $\M=\bfH+\bfP$. If any two of $\M$, $\bfH$ and $\bfP$ are $n\times n$ Hermitian matrices, then for $i=1,\ldots,n$:
\[
\lambda_i(\bfH)+ \lambda_n(\bfP) \leq \lambda_i(\M)\leq \lambda_i(\bfH)+ \lambda_1(\bfP).
\]
\end{lemma}

\begin{lemma}\label{lem0}
(Corollary 4.1 in \cite{vu2013minimax}). Let $\A$ be a $p\times p$ positive semidefinite matrix and suppose that its eigenvalues $\lambda_1(\A)\geq\ldots\geq\lambda_p(\A)$ satisfy $\lambda_d(\A)>\lambda_{d+1}(\A)$ for $d<p$. Let $\cE$ be the $d$-dimensional subspace spanned by the eigenvectors of $\A$ corresponding to its $d$ largest eigenvalues, and let $\E$ denote its orthogonal projector. Let $\cF$ be a $d$-dimensional subspace of $\bR^p$ and $\F$ be its orthogonal projector. If $\B$ is a symmetric matrix and $\F$ satisfies
\[
\langle \B,\E\rangle\leq\langle \B,\F\rangle,
\]
then
\[
\|\sin\bTheta(\cE,\cF)\|_F^2\leq\frac{\langle \B-\A,\F-\E\rangle}{\lambda_d(\A)-\lambda_{d+1}(\A)}.
\]
\end{lemma}

\begin{lemma}\label{lem:T1}
Let $\W_\bGamma=\bfS_Z-\bGamma$ and $\bPi=\V_0\V_0^\T$. Then
\[
\bP(\|\bPi \W_\bGamma\bPi\|_2\geq t) \leq 2\cdot 7^d \exp\left(-\frac{t^2/18}{18\alpha^2/n+\alpha t/n}\right).
\]
where $\alpha = \lambda_1\left(1+2c(q)^2R_q^{2/(2-q)}d/p\right)$.
\end{lemma}
\begin{proof}
Let $\cN_\delta$ be a minimal $\delta$-net of $\bfS_2^{d-1}$ for some $\delta\in(0,1)$. Then
\[
\|\bPi (\bfS_Z-\bGamma)\bPi\|_2 =\|\V_0^\T(\bfS_Z-\bGamma)\V_0\|_2 \leq (1-2\delta)^{-1}\max_{\bfu\in\cN_\delta}|\langle \V_0\bfu,(\bfS_Z-\bGamma)\V_0\bfu\rangle|
\]
by Proposition \ref{prop:net2}.
Note that
\begin{align*}
\langle \V_0\bfu,(\bfS_Z-\bGamma)\V_0\bfu\rangle &  = \frac{1}{n}\sum_{i=1}^n\left\{(\Z_i^\T\V_0\bfu)^2-\bE(\Z_i^\T\V_0\bfu)^2\right\}
\end{align*}
and
\begin{align*}
\|(\Z_i^\T\V_0\bfu)^2-\bE(\Z_i^\T\V_0\bfu)^2\|_{\psi_1} &\leq \|(\Z_i^\T\V_0\bfu)^2\|_{\psi_1} + 2\|\V_0^\T\bGamma \V_0\|_2 \\
& \leq \|\Z_i^\T\V_0\bfu\|_{\psi_2}^2 + 2\|\V_0^\T\bOmega \V_0\|_2 + 2\|\V_0^\T\bfH\V_0\|_2 \\
& \leq \|\bfT_i^\T\bOmega^{1/2}\G\V_0\bfu\|_{\psi_2}^2 + 2\lambda_1 +6c(q)^2\lambda_1R_q^{2/(2-q)}d/p \\
& \leq 3\lambda_1 + 6c(q)^2\lambda_1R_q^{2/(2-q)}d/p,
\end{align*}
where Proposition \ref{prop:ineq} is adopted in the second inequality, $\bfH = \bGamma-\bOmega$ and the bound for $\|\V_0^\T\bfH\V_0\|_2$ in Lemma \ref{lem:v0hv0} is plugged in the third step.
Then Lemma \ref{lem:bern} implies for all $t>0$ and $\bfu\in\cN_\delta$
\[
\bP(|\langle \V_0\bfu,(\bfS_Z-\bGamma)\V_0\bfu\rangle| \geq t) \leq 2\exp\left(-\frac{t^2/2}{2\alpha_1^2/n+\alpha_1 t/n}\right),
\]
where $\alpha_1 = 3\lambda_1 + 6c(q)^2\lambda_1R_q^{2/(2-q)}d/p$.
Choosing $\delta=1/3$ and applying Proposition \ref{prop:netnum}, we have $|\cN_\delta|\leq 7^d$ and 
\begin{align*}
\bP(\|\bPi (\bfS_Z-\bGamma)\bPi\|_2\geq t) & \leq \bP(\max_{\bfu\in\cN_\delta}|\langle \V_0\bfu,(\bfS_Z-\bGamma)\V_0\bfu\rangle|\geq t/3) \\
& \leq 2\cdot 7^d \exp\left(-\frac{t^2/18}{2\alpha_1^2/n+\alpha_1 t/(3n)}\right).
\end{align*}
\end{proof}

\begin{lemma}\label{lem:T2}
Let $\W_\bGamma=\bfS_Z-\bGamma$, $\bPi=\V_0\V_0^\T$ and $\bPi^\bot=\I-\bPi=\V_1^\T\V_1^\T$. Then
\[
\bP(\|\bPi^\bot \W_\bGamma\bPi\|_{2,\infty} \geq t) \leq 2p5^d\exp\left(-\frac{t^2/8}{2\beta^2/n+t\beta/n}\right),
\]
where $\beta = \sqrt{(\lambda_1\lambda_{d+1})} + 7c(q)\lambda_1R_q^{1/(2-q)}\sqrt{(d/p)}$.
\end{lemma}
\begin{proof}
Let $N_\delta$ be a minimal $\delta$-net in $\bfS_2^{d-1}$ for some $\delta\in(0,1)$ to be chosen later. By Proposition \ref{prop:net1} we have
\[
\|\bPi^\bot \W_\bGamma\bPi\|_{2,\infty} \leq \frac{1}{1-\delta}\max_{1\leq j\leq p}\max_{\bfu\in N_{\delta}}\langle\bPi^\bot \bfe_j,\W_\bGamma \V_0\bfu\rangle,
\]
where $\bfe_j$ is the $j$th column of $\I_{p\times p}$. Taking $\delta = 1/2$, by Proposition \ref{prop:netnum} we have $|N_{\delta}|\leq 5^d$.

Then
\[
\langle \bPi^\bot \bfe_j, (\bfS_Z-\bGamma)\V_0\bfu\rangle = \frac{1}{n}\sum_{i=1}^n\left\{\langle \Z_i,\bPi^\bot \bfe_j\rangle \langle \Z_i,\V_0\bfu\rangle - \bfe_j^\T\bPi^\bot \bfH \V_0\bfu\right\},
\]
where $\bfH=\bGamma-\bOmega$,
is the sum of independent random variables with mean zero. By Proposition \ref{prop:ineq} and the bound for $\|\bfH\V_0\|_2$ in (\ref{hv0}), the summands satisfy 
\begin{align*}
& ~\|\langle \Z_i,\bPi^\bot \bfe_j\rangle \langle \Z_i,\V_0\bfu\rangle - \bfe_j^\T\bPi^\bot \bfH\V_0\bfu\|_{\psi_1} \\
\leq &~\|\langle \Z_i,\bPi^\bot \bfe_j\rangle\|_{\psi_2}\|\langle \Z_i,\V_0\bfu\rangle\|_{\psi_2} + 2|\bfe_j^\T\bPi^\bot \bfH\V_0\bfu| \\
\leq &~\|\langle \bfT_i,\bOmega^{1/2}\G\bPi^\bot \bfe_j\rangle\|_{\psi_2} \|\langle \bfT_i,\bOmega^{1/2}\G\V_0\bfu\rangle\|_{\psi_2}  + 2\|\bfH\V_0\|_2\\
\leq &~\|\bfT_1\|_{\psi_2}^2\|\bOmega^{1/2}\G\bPi^\bot \bfe_j\|_2\|\bOmega^{1/2}\G\V_0\bfu\|_2 + 6c(q)\lambda_1R_q^{1/(2-q)}\sqrt{(d/p)}.
\end{align*}
Note that
\begin{enumerate}
\item[(i)] $\|\bOmega^{1/2}\G\V_0\|_2 \leq \sqrt{\lambda_1}$; 
\item[(ii)] $\|\bOmega^{1/2}\G\bPi^\bot\|_2 \leq \|\bOmega^{1/2}\G\V_1\|_2 = \|(\V_0\bLambda_0^{1/2}\V_0^\T+\V_1\bLambda_1^{1/2}\V_1^\T)\G\V_1\|_2$, \\
 $\|\V_0^\T\G\V_1\|_2 = \|p^{-1}\V_0^\T\J \V_1\|_2\leq c(q)R_q^{1/(2-q)}\sqrt{(d/p)}$, where the bound for $\|\J \V_0\|_2$ in Lemma \ref{lem:vjbound} is adopted,\\
$\Longrightarrow\|\bOmega^{1/2}\G\bPi^\bot\|_2\leq c(q)R_q^{1/(2-q)}\sqrt{(d\lambda_1/p)}+\sqrt{\lambda_{d+1}}$.
\end{enumerate}
We have 
\[
\|\langle \Z_i,\bPi^\bot \bfe_j\rangle \langle \Z_i,\V_0\bfu\rangle - \bfe_j^\T\bPi^\bot \bfH\V_0\bfu\|_{\psi_1} \leq \sqrt{(\lambda_1\lambda_{d+1})}+7c(q)\lambda_1R_q^{1/(2-q)}\sqrt{(d/p)} \triangleq \beta.
\]
Then Lemma \ref{lem:bern} implies that for all $t>0$ and every $\bfu\in \cN_\delta$
\begin{align*}
\bP(\|\bPi^\bot \W_\bGamma\bPi\|_{2,\infty}\geq t) & \leq \bP\left(\max_{1\leq j\leq p}\max_{\bfu\in N_\delta}\langle\bPi^\bot \bfe_j,\W_\bGamma \V_0\bfu\rangle\geq t/2\right) \\
& \leq p5^d\bP\left(|\langle\bPi^\bot \bfe_j,\W_\bGamma \V_0\bfu\rangle |\geq t/2 \right) \\
& \leq 2p5^d\exp\left(-\frac{t^2/8}{2\beta^2/n+t\beta/n}\right). 
\end{align*}
\end{proof}

\begin{lemma}\label{lem:bern}
(Bernstein's inequality). Let $\Y_1,\ldots,\Y_n$ be independent random variables with zero mean. Then
\[
\bP\left(\left|\sum_{i=1}^n\Y_i\right|>t\right)\leq 2\exp\left(-\frac{t^2/2}{2\sum_{i=1}^n\|\Y_i\|_{\psi_1}^2+t\max_{i\leq n}\|\Y_i\|_{\psi_1}}\right).
\]
\end{lemma}

%\begin{lemma}\label{lem:max}
%(Maximal inequality). Let $Y_1,\ldots,Y_n$ be arbitrary random variables that satisfy the bound
%\[
%\bP(|Y_i|>t)\leq 2\exp \left(-\frac{t^2/2}{b+at}\right)
%\]
%for all $t>0$, $i=1,\ldots,n$ and fixed $a,b>0$. Then
%\[
%\|\max_{i=1,\ldots,n}Y_i\|_{\psi_1}\leq c(a\log(1+n)+\sqrt{b\log (1+n)})
%\]
%for a universal constant $c>0$.
%\end{lemma}

\begin{proposition}\label{prop:eq} 
(Proposition C.1 in \cite{vu2013minimax}). If $\W$ is symmetric, and $\E$ and $\F$ are orthogonal projectors, then
\[
\langle \W,\F-\E \rangle = \langle \E^\bot \W\E^\bot,\F\rangle-\langle \E\W\E,\F^\bot\rangle+2\langle \E^\bot \W \E,\F\rangle.
\]
\end{proposition}

\begin{proposition}\label{prop:net1} 
(Proposition D.1 in \cite{vu2013minimax}). Let $\A$ be a $p\times d$ matrix, $(e_1,\ldots,e_p)$ be the canonical basis of $\bR^p$ and $\cN_\delta$ be a $\delta$-net of $\bfS_2^{d-1}$ for some $\delta\in[0,1)$. Then
\[
\|\A\|_{2,\infty}\leq(1-\delta)^{-1}\max_{1\leq j\leq p}\max_{\bfu\in\cN_\delta}\langle \bfe_j,\A \bfu\rangle.
\]
\end{proposition}

\begin{proposition}\label{prop:net2}
Let $\A$ be a $d\times d$ matrix, $\cN_\delta$ be a $\delta$-net of $\bfS_2^{d-1}$ for some $\delta\in [0,1)$. Then
\[
\|\A\|_2\leq (1-2\delta)^{-1}\max_{\bfu\in \cN_\delta}\langle \bfu,\A \bfu\rangle.
\]
\end{proposition}
\begin{proof}
There exist $\bfu_*\in \bfS^{d-1}$ and $\bfu\in \cN_\delta$ such that
\[
\|\A\|_2 = \langle \bfu_*,\A \bfu_*\rangle \text{ and } \|\bfu_*-\bfu\|_2\leq\delta.
\]
Then we have
\begin{align*}
\|\A\|_2 & =  \langle \bfu,\A \bfu\rangle +  \langle \bfu_*-\bfu,\A \bfu\rangle + \langle \bfu_*,\A(\bfu_*-\bfu)\rangle \\
& \leq \max_{\bfu\in \cN_\delta}\langle \bfu,\A \bfu\rangle + 2\delta\|\A\|_2
\end{align*}
Thus, 
\[
\|\A\|_2\leq (1-2\delta)^{-1}\max_{\bfu\in \cN_\delta}\langle \bfu,\A \bfu\rangle.
\]
\end{proof}

\begin{proposition}\label{prop:netnum} 
(Covering number of the sphere). Let $\cN_\delta$ be a minimal $\delta$-net of $\bfS_2^{d-1}$ for $\delta\in(0,1)$. Then
\[
|\cN_\delta|\leq(1+2/\delta)^d.
\]
\end{proposition}

\begin{proposition}\label{prop:ineq} 
(Proposition D.3 in \cite{vu2013minimax}). Let $X$ and $Y$ be random variables. Then
\[
\|XY\|_{\psi_1}\leq\|X\|_{\psi_2}\|Y\|_{\psi_2}.
\]
\end{proposition}

\section{More details for the alternating direction method of multipliers algorithms}\label{supp:algo}
\subsection{$V$-update of Algorithm \ref{admm1}}\label{supp:algo1}
For $q\in(0,1)$, the subproblems amount to
\begin{equation}\label{sub2-2}
\bfv = \arg\min_\bfv\frac{k_1}{2}\|\bfv\|_2^2+k_2\|\bfv\|_2^q+\bfv^\T\bxi\triangleq\argmin_\bfv f(\bfv),
\end{equation}
where $k_1=\beta+\rho$, $k_2=\alpha$ and $\bxi=\bfb_{i*},i=1,\ldots,p$. 
We take the derivative and set it to 0, and obtain
\begin{equation}\label{sub2-2-3}
k_1\bfv+k_2q\|\bfv\|_2^{q-2}\bfv+\bxi = 0 \quad \text{ if } \bfv\neq 0.
\end{equation}                                                                                
Let $x=\|\bfv\|_2>0$ and $k_3=\|\bxi\|_2$. We have 
\begin{equation}\label{sub2-equ}
k_1x+k_2qx^{q-1}=k_3.
\end{equation} 
The equation (\ref{sub2-equ}) has closed-form solutions for some $q$. For example, if $q=1/2$, then setting $z=\sqrt{x}$ leads to $k_1z^3-k_3z+k_2/2=0$; if $q=2/3$, then setting $z=x^{1/3}$ leads to $k_1z^4-k_3z+2k_2/3=0$. In both cases, we can obtain the analytic expressions of $x=\|\bfv\|_2$. Plugging it into (\ref{sub2-2-3}), we have the solution $\tilde \bfv$. The solution to (\ref{sub2-2}) is $\bfv=\argmin_{\bfv\in\{\tilde \bfv,0\}}f(\bfv)$. (If we cannot get a positive solution to (\ref{sub2-equ}), then $\bfv=0$ is the solution to (\ref{sub2-2}).)

\subsection{Algorithm for the column sparse principal subspace estimator}\label{supp:algo2}
%Now we consider the column sparse principal subspace estimator.
Analogously, the penalized version of (\ref{opt-col}) is
\begin{align}\label{opt-col-new}
\text{minimize }& -\langle \bfS_Z,\U\U^\T\rangle + \sum_{j=1}^d\alpha_j\|\bfv_{* j}\|_q^q + \frac{\mu}{2}\|\Y\|_F^2\\ \notag
\text{subject to }& \U\in\bV_{p,d} \\ \notag
&  \U-\V-\Y = 0,
\end{align}
where $\bfv_{* j}$ denotes the $j$th column of $\V$.
The augmented Lagrangian function for problem (\ref{opt-col-new}) is
\begin{align*}
\cT_\beta(\U,\V,\Y,\bLambda) = &-\langle \bfS_Z,\U\U^\T\rangle +\sum_{j=1}^d\alpha_j\|\bfv_{* j}\|_q^q + \frac{\mu}{2}\|\Y\|_F^2 + \langle \V-\U+\Y,\bLambda\rangle \\
& +\frac{\beta}{2}\|\V-\U+\Y\|_F^2,
\end{align*}
where $\bLambda$ is the Lagrange multiplier, $\beta>0$ is a penalty hyperparameter. Similarly, we linearize the objective function and define the approximated augmented Lagrangian function:
\begin{align*}
\widehat\cT_\beta^\U(\U;\widehat \U,\widehat \V,\widehat \Y,\bLambda) = & -\langle \bfS_Z,\widehat \U\widehat \U^\T\rangle +  \frac{\mu}{2}\|\widehat \Y\|_F^2 -2\langle \bfS_Z\widehat \U,\U-\widehat \U\rangle + \sum_{j=1}^d\alpha_j\|\widehat \bfv_{* j}\|_q^q \\ 
&+ \langle \widehat \V-\U+\widehat \Y,\bLambda\rangle +\frac{\beta}{2}\|\widehat \V-\U+\widehat \Y\|_F^2. 
\end{align*}
The linearized proximal alternating direction method of multipliers algorithm is described in Algorithm \ref{admm2}.
\vspace{-0.6cm}
\begin{algorithm}\label{admm2}
\caption{Linearized proximal alternating direction method of multipliers for column sparsity.}
\begin{tabbing}
  \qquad \enspace Input: Initial values $\U^0,\V^0,\Y^0,\bLambda^0$ and hyperparameters $\alpha_j,\beta,\mu,\rho$ \\
  \qquad \enspace For $k=0$ to $k=K-1$ \\
  \qquad\qquad  $\U^{k+1} = \arg\min\limits_{\U\in\bV_{p,d}} \widehat\cT_\beta^\U(\U;\U^k,\V^k,\Y^k,\bLambda^k)+\frac{\rho}{2}\|\U-\U^k\|_F^2$ \\
  \qquad\qquad  $\V^{k+1} = \arg\min\limits_\V\cT_\beta(\U^{k+1},\V,\Y^k,\bLambda^k)+\frac{\rho}{2}\|\V-\V^k\|_F^2$ \\
  \qquad\qquad  $\Y^{k+1} = \frac{1}{\mu+\beta}\left[\beta(\U^{k+1}-\V^{k+1})-\bLambda^{k+1}\right]$ \\
  \qquad\qquad  $\bLambda^{k+1} = \bLambda^k+\beta(\V^{k+1}-\U^{k+1}+\Y^{k+1})$ \\
  \qquad \enspace Output $\U^K$, $\V^K$
\end{tabbing}
\end{algorithm}

$\U$-update of Algorithm \ref{admm2} is the same as that of Algorithm \ref{admm1}. $\V$-update of Algorithm \ref{admm2} can be decoupled into one-dimensional subproblems:
\begin{align*}
& v_{ij} = \arg\min_{v_{ij}}\frac{\beta+\rho}{2}v_{ij}^2 + \alpha_j |v_{ij}|^q + b_{ij}v_{ij},~i = 1,\ldots,p, ~j = 1,\ldots,d,\text{ if } q\in(0,1], \\
& v_{ij} = \arg\min_{v_{ij}}\frac{\beta+\rho}{2}v_{ij}^2 + \alpha_j \text{I}(v_{ij}\neq 0) + b_{ij}v_{ij},~i = 1,\ldots,p, ~j = 1,\ldots,d,\text{ if } q=0,
\end{align*}
where $b_{ij}$ is the $(i,j)$ element of the matrix $\B = \bLambda^k+\beta(\Y^k-\U^{k+1})-\rho \V^k$ and $\text{I}(\cdot)$ is an indicator function. 
After some calculations, we obtain the solutions for $q = 1$ and $q = 0$.
\begin{proposition}
For $q=1$, the solution to $V$-update of Algorithm \ref{admm2} is $v_{ij} = S_{\alpha_j/(\beta+\rho)}(-b_{ij}/(\beta$ $+\rho))$, where $S_\lambda(y) = \max(|y|-\lambda,0)\sgn(y)$ is the soft thresholding function.
For $q=0$, the solution to $V$-update of Algorithm \ref{admm2} is $v_{ij} = -\text{I}(|b_{ij}|^2>2\alpha_j(\beta+\rho))b_{ij}/(\beta+\rho)$.
\end{proposition}
%The subproblems also have closed-form solutions for some $q\in(0,1)$, which is discussed in details in the Supplementary Material.
For $q\in(0,1)$, the subproblems amount to
\begin{equation}\label{col1-1}
v = \arg\min_v\frac{k_1}{2}v^2+k_2|v|^q+bv\triangleq\arg\min_vg(v),
\end{equation}
where $k_1=\beta+\rho$, $k_2=\alpha_j$ and $b=b_{ij},i=1,\ldots,p$, $j = 1,\ldots,d$.
We take the derivative and set it to 0, and obtain
\begin{equation}\label{col1-2}
k_1v+k_2q|v|^{q-2}v+b = 0 \quad \text{ if } v\neq 0.
\end{equation}                                                                                                         
Let $x=|v|>0$ and $k_3=|b|$. We have 
\begin{equation}\label{col1-equ}
k_1x+k_2qx^{q-1}=k_3.
\end{equation} 
The equation (\ref{col1-equ}) has closed-form solutions for some $q$. For example, if $q=1/2$, then setting $z=\sqrt{x}$ leads to $k_1z^3-k_3z+k_2/2=0$; if $q=2/3$, then setting $z=x^{1/3}$ leads to $k_1z^4-k_3z+2k_2/3=0$. In both cases, we can obtain the analytic expressions of $x=|v|$. Plugging it into (\ref{col1-2}), we have the solution $\tilde v$. The solution to (\ref{col1-1}) is $v=\argmin_{v\in\{\tilde v,0\}}g(v)$. (If we cannot get a positive solution to (\ref{col1-equ}), then $v=0$ is the solution to (\ref{col1-1}).) 

For the column sparsity, to simplify the tuning procedure of $\alpha_j$, we let $\alpha_j=\alpha/\|\bfv^0_{*j}\|_1$, where $\bfv_{*j}^0$ is the $j$th column of $\V^0$, $j=1,\ldots,d$, and select $\alpha$ by 5-fold cross-validation. 

\subsection{Selection of $\mu$}\label{supp:mu}

\begin{figure}[htb!]
  \centering
  \includegraphics[width=0.49\textwidth]{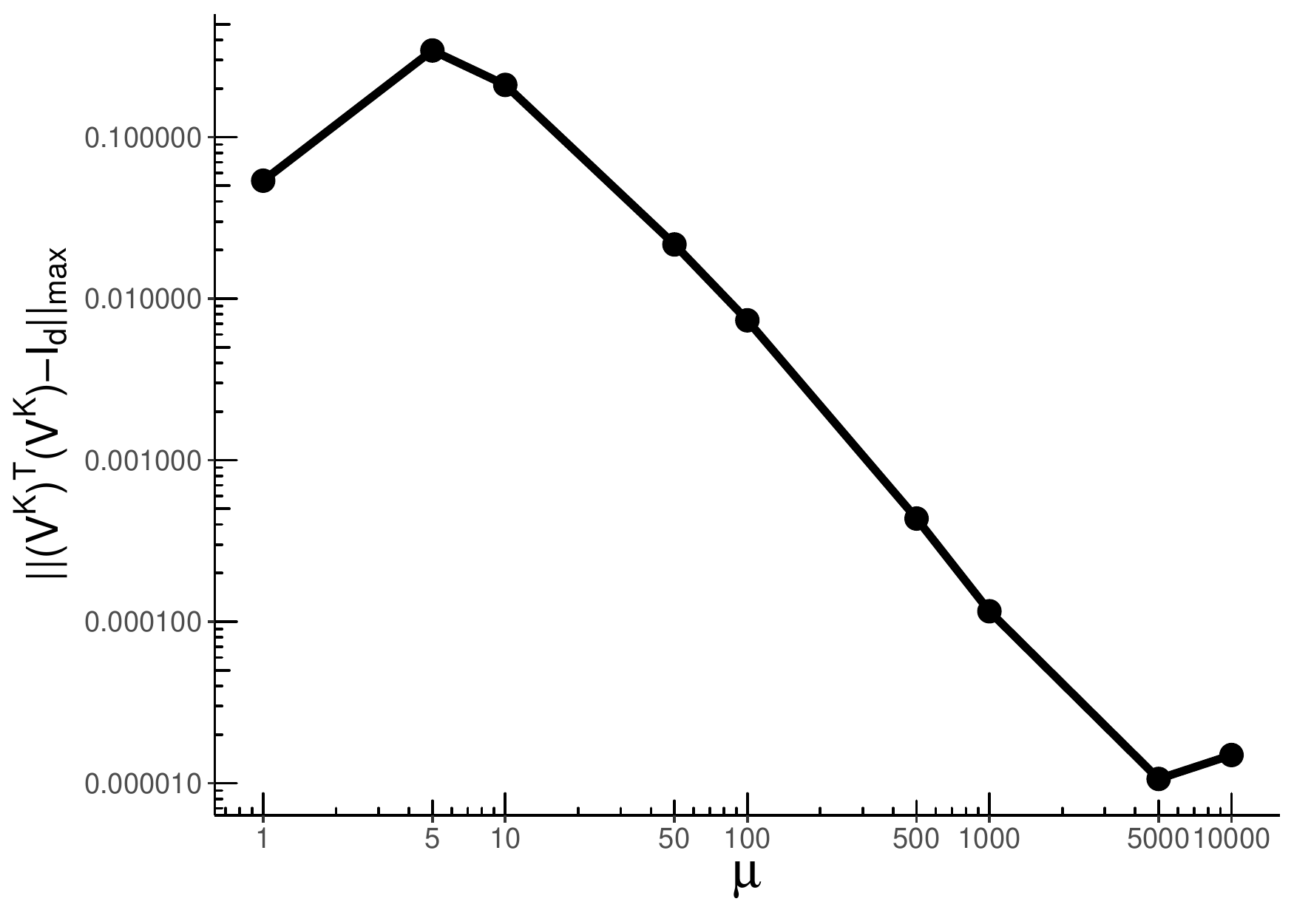}
  \includegraphics[width=0.49\textwidth]{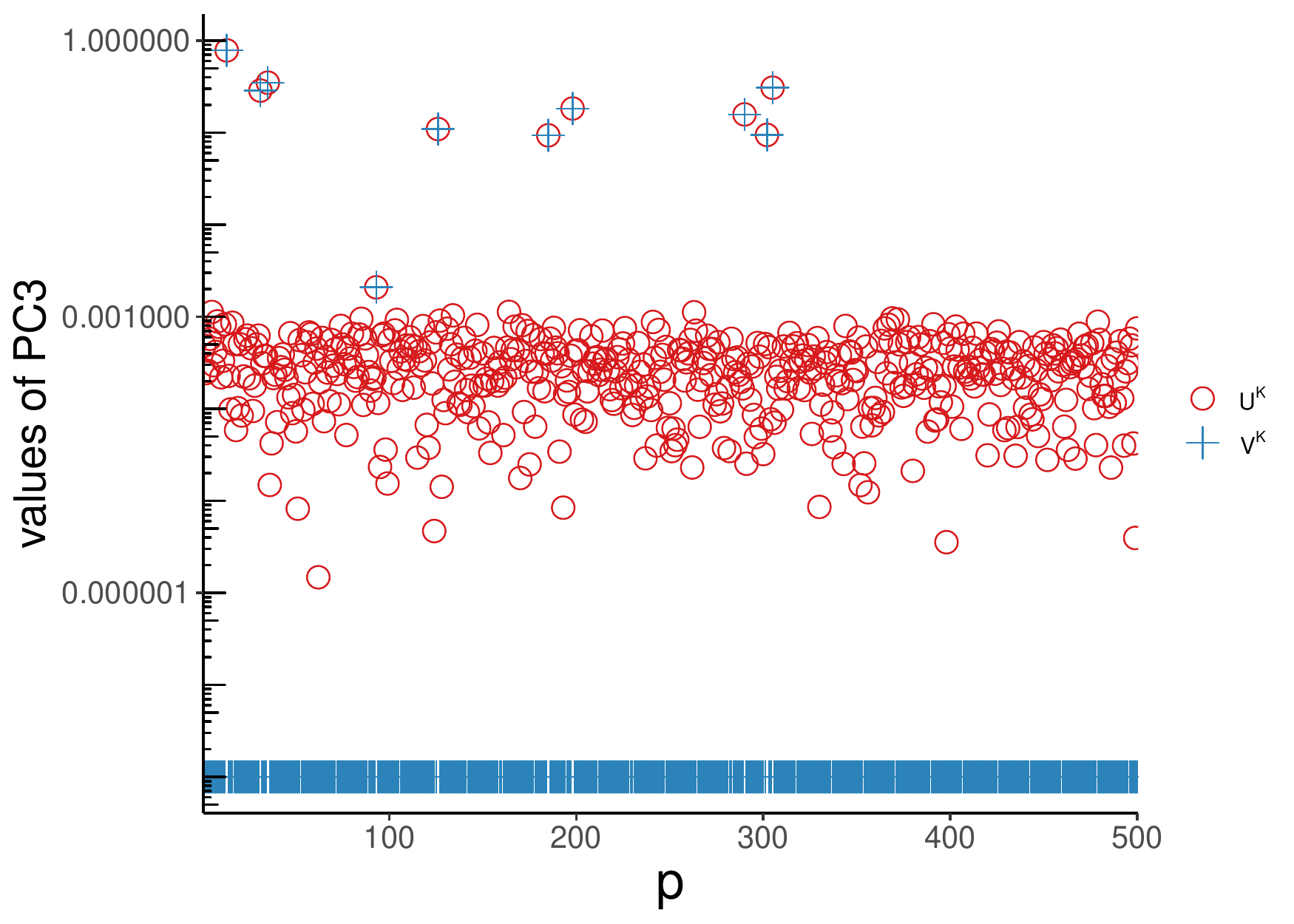}
  \caption{Under $n=250$ and $p=500$, the output of one randomly chosen simulation by the proposed method with row sparsity $q=0$. Left: the values of $\|(\V^K)^\T(\V^K)-\I_d\|_{\max}$ with varying $\mu$. Right: the values of the third column of the output $\U^K$ and $\V^K$ when $\mu=1000$.} \label{fig:UV}
  \end{figure}
  
We randomly choose one simulation setting under $n=250$ and $p=500$ with the row sparsity $q=0$, to examine the influence of the hyperparameter $\mu$. Shown in the left panel of Fig.\ \ref{fig:UV} is how $\| (\V^K)^\T(\V^K)-\I_5\|_{\max}$ changes with  $\mu$. Since $\|(\V^K)^\T(\V^K)-\I_5\|_{\max}=1.16\times 10^{-4}$ when $\mu=1000$, $\V^K$ is nearly orthonormal. The right panel of Fig.\ \ref{fig:UV} shows the difference between $\U^K$ and $\V^K$ under $\mu=1000$. Take a look at the values of the output $\U^K$. Although the important features pop out, one still needs to set a threshold carefully to identify them.
By comparison, $\V^K$ is nearly orthonormal and sparse, whose nonzero values automatically correspond to the important features. Hence we take $\V^K$ as our estimator.

\section{Additional results of the real application}\label{supp:real}
Table \ref{tab:appcol1} lists the words selected under the column sparsity with $q=1$.
The biplots of the first two principal components under the row sparsity with $q=1$ are shown in Fig.\ \ref{fig:realrow1}.

\begin{table}
\def~{\hphantom{0}}
\caption{Results under the column sparsity with $q=1$. Denote the proposed, Log, Raw and Power methods by M1, M2, M3 and M4, respectively. Denote by PC1 and PC2 the first and second principal components}{
%\vspace{0.2cm}
\begin{tabular}{@{}llrrrrlrrrr@{}}
     & Word & M1 & M2 & M3 & M4 & Word & M1 & M2 & M3 & M4 \\ [5pt]
 PC1 & bias & 0.15 &  &  &  & normal & 0.11 &  &  & 0.02\\
     & covariance & 0.40 & 0.51 & 0.13 & 0.53 & number &  &  & 0.04 & \\
     & coefficient & 0.22 &  & 0.15 & 0.08 & oracle &  &  & 0.06 & \\
     & composition &  &  & 0.06 &  & penalty &  &  & 0.22 & \\
     & continuous &  &  & -0.01 &  & point &  &  & -0.04 & -0.01 \\
     & criterion &  &  & 0.01 &  & process &  &  & -0.40 & -0.02 \\
     & equal & 0.21 &  &  & 0.01 & property &  &  & 0.02 & \\
     & generalize &  &  & 0.02 & 0.01 & random &  &  & -0.01 &  \\
     & group &  &  & 0.18 &  & regression & 0.36 & 0.53 & 0.35 & 0.75 \\
     & hierarchy & -0.02 &  &  &  & response &  &  &  & 0.01 \\
     & high &  &  &  & -0.05 & select & 0.03 &  & 0.50 & 0.06 \\
     & inference & 0.03 &  &  &  & semiparametric & 0.48 & 0.61 &  & 0.13 \\
     & lasso &  &  & 0.11 &  & space &  &  & -0.22 & -0.02 \\
     & likelihood & 0.39 & 0.23 & 0.04 & 0.14 & sparse &  &  & 0.01 &  \\
     & linear & 0.26 & 0.13 & 0.12 & 0.30 & test &  &  & -0.02 & \\
     & maximal & 0.19 &  &  &  & time &  &  & -0.50 & \\
     & missing & 0.08 &  &  &  & volatilization &  &  & -0.02 & \\
     & nonparametric & 0.25 &  &  & 0.09 &  &  &  &  & \\ [5pt]
 PC2 & adaptive & -0.04 &  &  &  & null &  &  & 0.03 &  \\    
     & baseline & 0.18 &  &  &  & number &  &  & -0.04 & -0.27 \\    
     & Bayes &  &  & -0.01 &  & optimize &  &  &  & -0.01 \\    
     & cancer & 0.10 &  &  &  & predict &  &  & -0.23 &  \\    
     & censor & 0.50 & 0.58 &  &  & process &  &  &  & 0.25 \\    
     & classify &  &  & -0.23 &  & proportion & 0.19 &  &  &  \\    
     & clinic & 0.10 &  &  &  & select &  &  &  & -0.18 \\    
     & cluster &  &  & -0.18 &  & smooth & -0.04 &  &  &  \\    
     & composition &  &  &  & -0.03 & space & -0.03 &  & -0.01 &  \\    
     & cumulative & 0.07 &  &  &  & sparse & -0.13 &  &  &  \\    
     & dimension & -0.16 &  &  & -0.54 & statistic &  &  & 0.18 &  \\    
     & disease & 0.06 &  &  &  & survive & 0.50 & 0.63 &  &  \\    
     & equal &  &  & 0.01 &  & test &  &  & 0.86 &  \\    
     & hazard & 0.49 & 0.50 &  &  & time & 0.25 &  &  & 0.66 \\    
     & high & -0.15 &  &  & -0.33 & transform & 0.12 &  &  &  \\    
     & hypothesis &  &  & 0.02 &  & value &  &  & 0.04 &  \\    
     & likelihood &  &  & 0.28 &  &  &  &  &  &  \\    
\end{tabular}}\label{tab:appcol1}
\end{table}

\begin{figure}[htb!]
   \centering
   \includegraphics[width=0.95\textwidth]{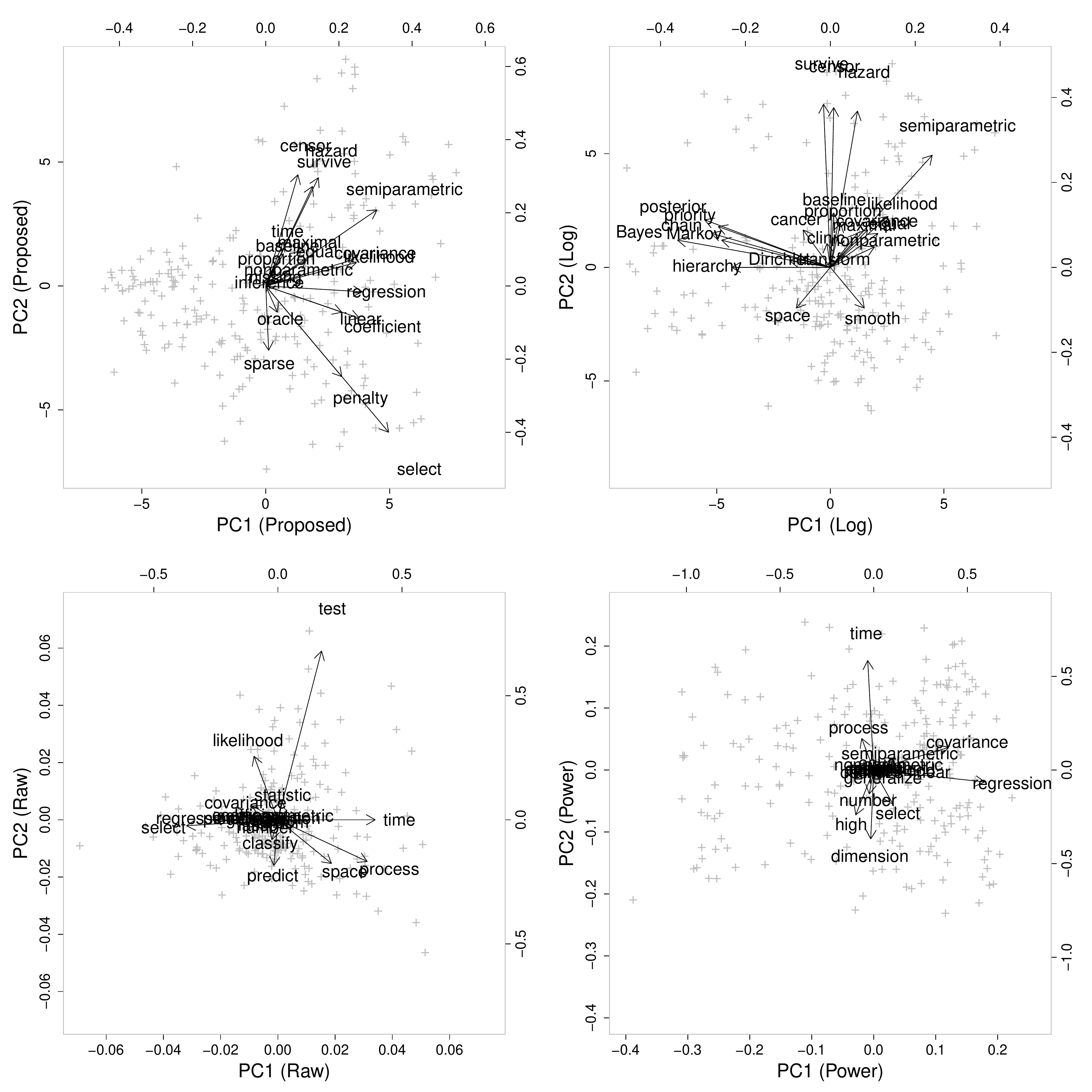}
   \caption{Biplots of the first two principal components for the proposed, Log, Raw and Power methods under the row sparsity with $q=1$.}\label{fig:realrow1}
\end{figure}

\end{document}